\documentclass[showpacs,amsmath,amssymb,onecolumn,superscriptaddress,notitlepage,preprintnumbers,pra]{revtex4-1}

\usepackage{graphicx}
\usepackage{dcolumn}
\usepackage{bm}
\usepackage{times}
\usepackage{multirow}
\usepackage{dcolumn}
\usepackage{tabularx}
\usepackage{subfigure}
\usepackage{ae}
\usepackage{lettrine}
\usepackage{color}
\usepackage{amsmath,amsthm,amssymb,amsfonts,boxedminipage}
\usepackage{epstopdf}
\usepackage{hyperref}
\usepackage{dsfont}
\newcommand{\algmargin}{\the\ALG@thistlm}
\makeatother

\newtheorem{theorem}{Theorem}
\newtheorem{lemma}{Lemma}
\newtheorem{corollary}{Corollary}
\newtheorem{definition}{Definition}
\newtheorem{assumption}{Assumption}

\makeatletter
\renewcommand\@biblabel[1]{#1.}
\makeatother

\usepackage[noend,noline,ruled,linesnumbered]{algorithm2e}
\usepackage{hyperref}
\usepackage{bm}
\usepackage{times}
\usepackage{multirow}
\usepackage{dcolumn}
\usepackage{tabularx}
\usepackage{graphicx}
\usepackage{subfigure}
\usepackage{ae}
\usepackage{lettrine}
\usepackage{color}
\usepackage{makecell}

\usepackage{amsmath,amsthm,amssymb,amsfonts,boxedminipage}
\usepackage{epstopdf}

\definecolor{comment}{rgb}{0, 0, 0}

\hypersetup{colorlinks=true, linkcolor=cyan, citecolor=cyan, urlcolor=blue }

\usepackage{braket}
\newcommand{\tr}[1]{\text{Tr}\left( #1 \right)}

\newcommand{\abs}[1]{\left| #1 \right|}

\newcommand{\poly}{\mathrm{poly}}

\newcommand{\PKU}{Center on Frontiers of Computing Studies, School of Computer Science, Peking University, Beijing 100871, China}
\newcommand{\bnu}{School of Artificial Intelligence,
 Beijing Normal University, Beijing,
 100875, China}

\begin{document}


\title{A Dequantized Algorithm for the Guided Local Hamiltonian Problem}

\author{Yukun Zhang}
\thanks{These authors contributed equally.}
\affiliation{\PKU}

\author{Yusen Wu}
\thanks{These authors contributed equally.}
\affiliation{\bnu}

\author{Xiao Yuan}
\email{xiaoyuan@pku.edu.cn}
\affiliation{\PKU}

\date{\today}

\begin{abstract}
The local Hamiltonian (LH) problem, the quantum analog of the classical constraint satisfaction problem, is a cornerstone of quantum computation and complexity theory. It is known to be QMA-complete, indicating that it is challenging even for quantum computers. Interestingly, the guided local Hamiltonian (GLH) problem~---~an LH problem with a guiding state that has a non-trivial overlap with the ground state~---~can be efficiently solved on a quantum computer and is proved to be BQP-complete. This makes the GLH problem a valuable framework for exploring the fundamental separation between classical and quantum computation.
Remarkably, the quantum algorithm for solving the GLH problem can be `dequantized' (i.e., made classically simulatable) under certain conditions, such as when only constant accuracy is required and when the Hamiltonian satisfies an unrealistic constant operator norm constraint. In this work, we relieve these restrictions by introducing a dequantized classical algorithm for a randomized quantum imaginary-time evolution quantum algorithm. We demonstrate that it achieves either limited or arbitrary constant accuracy, depending on whether the guiding state’s overlap is general or exceeds a certain threshold. Crucially, our approach eliminates the constant operator norm constraint on the Hamiltonian, opening its applicability to realistic problems.
Our results advance the classical solution of the GLH problem in practical settings and provide new insights into the boundary between classical and quantum computational power.

\end{abstract}

\maketitle



The ground-state problem has emerged as one of the most significant challenges in science, owing to its profound importance in both physics and computer science. In quantum physics, it is pivotal for understanding exotic phenomena associated with low-lying energy states, such as superconductivity~\cite{wilson1983superconducting}, superfluidity~\cite{wheatley1975experimental}, and topological orders~\cite{wen1995topological,kane2005z}. These studies have broad applications across fields ranging from condensed matter physics~\cite{wheatley1975experimental,wilson1983superconducting,wen1995topological,kane2005z} and high-energy physics~\cite{kovchegov2013quantum} to quantum chemistry~\cite{levine2009quantum,cao2019quantum}.
In computer science, the problem is closely related to the local Hamiltonian (LH) problem~\cite{kempe2006complexity,gharibian2015quantum}, which has been extensively studied over the past two decades. The LH problem is proven to be QMA-complete (analogue to NP-complete and hence unlikely to be efficiently solvable on a quantum computer) for both synthetic~\cite{kitaev2002classical,kempe2006complexity} and physical quantum systems~\cite{childs2014bose,o2021electronic}, making it the quantum analog~\cite{aharonov2002quantum} of classical constraint satisfaction problems that characterize the complexity class NP, as established by the celebrated Cook-Levin theorem~\cite{karp2010reducibility}.
Furthermore, the quantum PCP conjecture~\cite{aharonov2013guest}, which posits that verifying the ground-state energy of an LH problem to within constant precision is QMA-hard in the worst case, remains one of the most profound and unresolved questions in quantum complexity theory over the past decade.

The search for exponential quantum advantages in solving the LH problem has been a central motivation and enduring pursuit in the study of quantum computation. Remarkably, the LH problem becomes efficiently solvable on a quantum computer when provided with a guiding state that has a non-trivial overlap with the ground state. Significant progress~\cite{lin2020near,dong2022ground,lin2022heisenberg,wan2022randomized,wang2023quantum,ding2023even,ni2023low} has been made in recent years. Near-optimal quantum algorithms~\cite{lin2020near,dong2022ground} have been developed, and specialized quantum algorithms tailored for early fault-tolerant quantum devices have been introduced, achieving substantial reductions in query depth~\cite{lin2022heisenberg,wan2022randomized,wang2023quantum,ding2023even,ni2023low}. 
In computer science, the LH problem with a non-trivial guiding state is referred to as the guided local Hamiltonian (GLH) problem~\cite{gharibian2022dequantizing}.
The GLH problem can also be viewed as a decision version of the ground-state energy estimation (GSEE) problem, where the goal is to approximate the ground-state energy to a specified accuracy.
This problem is proven to be BQP-complete~\cite{cade2022improved, cade2022complexity}, meaning it represents some of the “hardest” problems that can be efficiently solved on a quantum computer. These theoretical results suggest that exponential quantum advantage may be achievable for quantum systems where a reasonable guiding state is available.

However, the situation is more nuanced and delicate than it first appears. 
The GLH problem may become classically solvable, negating exponential quantum advantages when certain relaxed conditions are imposed. Specifically, it has been shown that the GLH problem can be efficiently solved classically if the guiding state is classically accessible, the required accuracy for estimating the ground-state energy is constant, and a constant bounds the operator norm of the Hamiltonian, $\|H\| \leq 1$ ~\cite{gharibian2022dequantizing,gall2024classical}. This conclusion was reached by constructing a dequantized classical algorithm based on the recently proposed near-optimal quantum algorithms~\cite{lin2020near, gilyen2019quantum} that efficiently solves the GLH problem under these conditions. Similar dequantization techniques have been explored for various quantum algorithms; see Ref.~\cite{tang2022dequantizing} for a comprehensive review.
However, the restriction on the operator norm of the Hamiltonian limits the applicability of the classical algorithm to most physical systems, including even the simplest Ising Hamiltonian, where $\|H\| = \mathcal{O}({\rm poly}(n))$, with $n$ representing the system size. While this constraint has been removed for 2D Hamiltonians~\cite{wu2024efficient}, finding an efficient classical algorithm for general and realistic GLH problems under mild and physically reasonable assumptions remains an open challenge.

Here, we address this challenge by presenting an efficient classical algorithm that dequantizes a novel quantum approach, termed randomized quantum imaginary-time evolution (RQITE), to solve the GLH problem. The RQITE algorithm shares structural similarities with recently proposed early fault-tolerant quantum algorithms but introduces an imaginary-time evolution filter~---~an exponential operator acting on the Hamiltonian. This construction allows for the use of cluster expansion to dequantize RQITE, providing an efficient classical solution to the GLH problem for general Hamiltonians without the operator norm constraint.
Our method draws inspiration from recent advancements in classical simulation of Hamiltonian evolution~\cite{wild2023classical,wu2024efficient} and Gibbs state preparation~\cite{yin2023polynomial,mann2024algorithmic,bakshi2024high}. However, the evolution time in the dequantized RQITE algorithm remains below a constant threshold, limiting it to only approximate the ground-state energy with finite accuracy. Extending beyond this threshold is theoretically infeasible, as the hardness of approximation results for certain Gibbs states~\cite{sly2012computational,goldberg2017complexity,mann2024algorithmic} demonstrate that classical simulation becomes NP-hard for larger evolution times.
Despite this limitation, we achieve arbitrary constant accuracy in energy estimation by leveraging analytic continuation techniques~\cite{wild2023classical} and further assuming a lower bound on the overlap of the guiding state with the ground state. Such an assumption has also been studied and utilized in quantum algorithms for ground-state energy estimation and quantum phase estimation~\cite{ding2023even,ni2023low}, where the query depth decreases as the overlap increases.
Finally, we discuss the implications of these findings for understanding the boundary between classical and quantum computation, shedding light on the potential and limitations of quantum computing in this domain. 

In the following, we focus on solving the GSEE and, consequently, the GLH problems under several assumptions. The main results are summarized in the main text, with detailed explanations provided in the Supplementary Materials.

\vspace{0.2cm}
\noindent\textbf{\emph{Randomized quantum imaginary-time evolution algorithm (RQITE).---}}
We first review the GSEE problem and introduce the RQITE algorithm for the GSEE problem. 

\begin{figure*}
    \centering
    \includegraphics[width=1\textwidth]{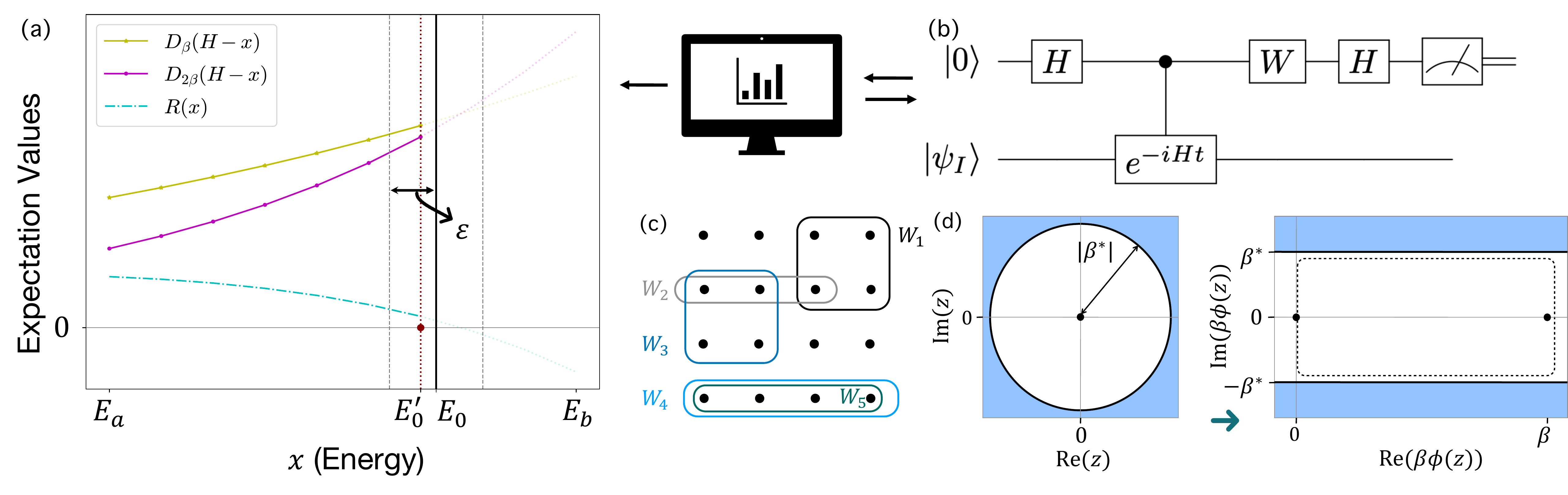}
    \caption{(a) Sketch of the $D_\beta(H-x)$, $D_{2\beta}(H-x)$ and $R(x)$ with $x$ evaluated at an interval $\varepsilon$. When close to $E_0$, $R(x)$ decreases monotonically. When $R(x)$ at $E_0^\prime$ is sufficiently small, the algorithm terminates and outputs $E_0^\prime$. (b) The Hamdard test circuit with the controlled operation as a Hamiltonian evolution operator. A classical computer determines the evolution time, and measurement results of the circuit are post-processed to estimate the partition functions. (c) Sketch of a cluster $\bm W$. The disconnected cluster consists of two connected clusters $\bm W=\bm{W}_a \cup \bm{W}_b$, where $\bm{W}_a=\{\bm{W}_1,\bm{W}_2,\bm{W}_3\}$ and $\bm{W}_b=\{\bm{W}_4,\bm{W}_5\}$. (d) Sketch of the analytic continuation approach. The function $\phi(z)$ is analytic for $\abs{z}\leq \beta^*$. The analytic continuation implements the map for a given $\beta\in\mathbb{R}$ to $\beta\phi(z)$, resulting in the elongated region.}
    \label{fig:main}
\end{figure*}

We consider $k$-local Hamiltonian that can be described as $H=\sum_{X\in S}\lambda_Xh_X$, where $S$ is a subsystem, each term $h_X$ acts nontrivially on subsystem $X$ of constant size $\leq k$, and $\lambda_X\in\mathbb{R}$ is the corresponding coefficient. Here, we suppose $\|h_X\|=1$ and $|\lambda_X|\leq 1$ for every $X\in S$ without loss of generality. Note that the operator norm of the Hamiltonian $\|H\|$ generally scales polynomially with the number of qubits.
Meanwhile, we utilize the related interaction graph $G$~\cite{wild2023classical,haah2024learning} to characterize the locality of $H$. The vertices of $G$ correspond to each term in the Hamiltonian. An edge exists between two vertices if and only if the qubits acted upon by the corresponding terms overlap. We denote the maximal degree of the interaction graph by $\mathfrak{d}=\max_{h_X\in H}\left\{\mathfrak{d}(h_X)\right\}$. Given these notations, we can formally introduce the GSEE problem.
\begin{definition}[The Ground State Energy Estimation (GSEE) problem]
\label{def:gsee}
Consider an $n$-qubit (local) Hamiltonian $H$, let $E_0<E_1\leq\cdots\leq E_{N-1}$, where $N=2^n$, be eigenvalues of $H$ with corresponding eigenstates $\ket{\psi_0},\ket{\psi_1},\cdots,\ket{\psi_{N-1}}$. Suppose a lower-bounded estimation $\Delta>0$ for the energy gap is given such that $E_1-E_0\geq \Delta$. Suppose we can also prepare an initial state $\ket{\psi_I}$ such that its overlap with the ground state satisfies $|\langle \psi_I| \psi_0\rangle|\geq \gamma$. The GSEE task estimates $E_0^\prime$ for the ground-state energy such that $|E_0-E_0^\prime|\leq\varepsilon$.
\end{definition}
The RQITE algorithm is inspired by recent randomized Fourier estimation (RFE) algorithms~\cite {lin2022heisenberg,wan2022randomized,wang2023quantum} that estimate $E_0$ by analyzing the convolution function using a Hadamard test circuit. The convolution function is defined as 
\begin{equation}
\begin{aligned}
    C(x)&:=(P\ast F)(x)=\sum_{j=0}^{N-1} p_j \int_{-\infty}^\infty\delta(\tau-E_j)  F(x-\tau)d\tau\\
    &=\sum_{j=0}^{N-1} p_j F(x-E_j),
\end{aligned}
\end{equation}
where $P(x):=\sum_{j=0}^{N-1} p_j \delta(\tau-E_j)$ is the spectral function of the guiding state, $p_j:=|\langle \psi_j|\psi_I \rangle|^2$, $\delta(\cdot)$ is the Dirac Delta function, and $F(x)$ is function that to be appointed. The convolution function is the linear combination of spectral weighted $F(x)$ located at each eigenenergy. As such, we may estimate $E_0$ by selecting functions $F(x)$ that behave distinctively at $x=0$. Then, using the lower-bounded overlap information $\gamma$, we can take the smallest value approaches the behavior at $x=0$ as an approximation to $E_0$ in an interval $[E_a,E_b]$ promised to contain $E_0$. For instance, Ref.~\cite{lin2022heisenberg} applies the Heaviside function, which leaps at the origin. This also makes $C(x)$ the cumulative distribution function.

For the RQITE, we first define a partition function: 
\begin{equation}\label{eq:partition_func}
\begin{aligned}
    D_\beta(H-x):=& \bra{\psi_I} e^{-\beta (H-x)}\ket{\psi_I}\\
    =&\sum_{j=0}^{N-1} p_j e^{-\beta (E_j-x)}=\sum_{j=0}^{N-1} p_j F(x-E_j),
\end{aligned}
\end{equation}
where $F(x)=e^{\beta x}$ with $\beta$ the parameter, $(\beta-x)\equiv (\beta-x\cdot\mathds{1})$ and we omit the identity $\mathds{1}$ for simplicity. The equivalence of the partition and convolution function can be easily checked. See Theorem~\ref{theorem:partition_convolution_equivalence} for a more rigorous proof. The partition function can be seen as the (square of) normalization factor for the imaginary-time evolved state $e^{-\beta(H-x)/2}\ket{\psi_I}$. Yet, we could not distinguish the ground-state energy from $D_\beta(H-x)$. To this end, we define the residue function as 
\begin{equation}\label{eq:partition_difference}
    R(x):= D_\beta(H-x) - D_{2\beta}(H-x).
\end{equation}
As shown in Fig.~\ref{fig:main}~(a), $R(x)$ decades monotonically to zero when $x\leq E_0$ approaches $E_0$. Hence, our scheme to pinpoint the ground-state energy is to estimate $R(x)$ beginning at $x=E_a$ and an interval of $\varepsilon$. The algorithm outputs the estimation $E_0^\prime$ when $R(x)$ decays below the termination threshold $\Xi$ (details refer to Theorem~\ref{thm:rqite}) that is sufficiently close to zero and outputs the estimation $E_0^\prime$. A key requirement for the algorithm is that the spectral gap is larger than the accuracy. The rationale is that around $x=E_0$, the effect of the first-excited state component in $R(x)$ can be controlled by $e^{-\beta \Delta}$ so that it will not affect the determination of $E_0^\prime$ dramatically. Specifically, we assume that
\begin{equation}\label{eq:gap_assumption}
    \frac{\Delta}{\varepsilon}\geq \ln(\gamma^{-2}\varepsilon^{-1}),
\end{equation}
a requirement that is satisfiable by various quantum molecular systems ranging from medium to relatively large system size~\cite{wang2023quantum}. We provide more detailed explanations in the Assumption~\ref{assume:gap_accuracy_relation} in the Appendix.

The estimation of the partition (convolution) function follows the standard procedure in RFE algorithms~\cite{lin2022heisenberg,wan2022randomized,wang2023quantum} by inspecting the Fourier transformation of 
\begin{equation}\label{eq:exp_fourier}
    e^{-\beta(H-x)}=\int_{-\infty}^\infty\mathrm{~d}t \hat{g}(t) e^{- i t(H-x) },~(H-x)\succeq0,
\end{equation}
where $\hat{g}(t)=\frac{\beta}{\pi\left(\beta^2+t^2\right)}$ is the Cauchy-Lorentz distribution. Here, the $(H-x)\succeq0$ results from that, without loss of generality, we assume that the algorithm terminates at $x\leq E_0$. As such, the function can be equivalently viewed. We then truncate the integral to approximate the partition function
\begin{equation}\label{eq:trunc_exp_fourier}
\begin{aligned}
    D_\beta(H-x)&\approx \int_{-T}^T \mathrm{~d}t\frac{\beta}{\pi\left(\beta^2+t^2\right)}  \bra{\psi_I}e^{- i (H-x)t}\ket{\psi_I}\\
    &=\mathbb{E}_{t\sim\hat{g}_T^\prime(t)} [ \mathbf{Z}(x)]\equiv \|\hat{g}_T\|e^{i xt}(\mathbb{E}\left[\mathbf{X}(t)\right]+\mathbb{E}\left[\mathbf{Y}(t)\right]),
\end{aligned}
\end{equation}
where $\hat{g}_T(x)$ is the Cauchy-Lorentz distribution truncated to the domain $[-T,T]$, $\hat{g}_T^\prime(x)=\hat{g}_T(x)/\|\hat{g}_T\|$, $\|\hat{g}_T\|$ is the normalization factor, $\mathbf{Z}(x):=\|\hat{g}_T\|\bra{\psi_I}e^{- i (H-x)t}\ket{\psi_I}$, $\mathbb{E}\left[\mathbf{X}(t)\right]={\rm{Re}}\left({\rm tr}\left[\rho_I e^{-i Ht}\right]\right)$, and $\mathbb{E}\left[\mathbf{Y}(t)\right]={\rm{Im}}\left(\mathrm{tr}\left[\rho_I e^{-i Ht}\right]\right)$ are real and imaginary estimators with $\rho_I=\ket{\psi_I}\bra{\psi_I}$. We thus estimate the partition function by sample $t$ from $\hat{g}_T(x)$ and for each $t$ we evaluate the real and imaginary estimator using the Hadamard circuit as shown in Fig.~\ref{fig:main} (b) with $W$ set to $I$ or $S^\dagger$. In the Appendix, we meticulously analyze both the parameter setting of the algorithm and its complexity, which we summarized in the following theorem.
\begin{theorem}\label{thm:rqite}
Given the GSEE problem given by Definition~\ref{def:gsee} and condition in Eq.~\eqref{eq:gap_assumption}, the RQITE algorithms can estimate the ground-state energy within additive error $\varepsilon$ by setting $\beta=\frac{\ln(\gamma^{-2}\varepsilon^{-1})}{\Delta}$, $T=\frac{4}{\pi\gamma^2\varepsilon}$, and termination threshold $\Xi=\left(\frac{\beta}{2}+1\right)\gamma^2\varepsilon$. The cost of the algorithm is given by
\begin{itemize}
    \item The maximal Hamiltonian evolution time is $t_{\text{max}}=\mathcal{O}(\varepsilon^{-1}\gamma^{-2})$;
    \item The total Hamiltonian evolution time is $t_{\text{total}}=\mathcal{\widetilde{O}}(\varepsilon^{-3}\Delta^2\gamma^{-6})$;
    \item The classical postprocess time is $\mathcal{\widetilde{O}}(\varepsilon^{-4}\Delta^2\gamma^{-6})$.
\end{itemize}
\end{theorem}

\noindent We note that the complexity is not optimal with respect to these parameters in general; however, it is sufficient for our dequantization algorithm.

\begin{table*}
    \footnotesize
    \centering
    \caption{Comparisons of our results with related previous studies on solving the GLH problem, focusing on accuracy, input constraints, and computational time complexity.}\label{Table}
    \begin{tabular}{c|cccc}
      \hline\hline
      \textbf{Algorithm} & \textbf{Method} &\textbf{Accuracy} & \textbf{Constraints} &
      \textbf{Complexity}\\
      \hline\multirow{3}{*}{Refs~\cite{cade2022improved, cade2022complexity}} & \multirow{3}{*}{---}& \multirow{3}{*}{$\varepsilon=\frac{1}{{\rm poly}(n)}$} & \multirow{3}{*}{\makecell[c]{$k$-local, $\|H\|\leq 1$, \\$\gamma\in(\frac{1}{{\rm poly}(n)},1-\frac{1}{{\rm poly}(n)})$}} & \multirow{3}{*}{BQP-complete} \\ \\ \\
      \hline\multirow{2}{*}{Ref~\cite{gharibian2022dequantizing}} & \multirow{2}{*}{Dequantized QSVT} & \multirow{2}{*}{$\varepsilon=\mathcal O(1)$} & \multirow{2}{*}{$s$-sparse, $\|H\|\leq 1$} & \multirow{2}{*}{ $\mathcal{O}\left(\gamma^{-4}(\abs{S}2^s+1)^{(2+4/\varepsilon)\log(1/\gamma)}\right)$}\\ \\
      \hline\multirow{2}{*}{Ref~\cite{gall2024classical}} & \multirow{2}{*}{Dequantized QSVT} & \multirow{2}{*}{$\varepsilon\|H\|=\mathcal O(\|H\|)$} & \multirow{2}{*}{$k$-local} & \multirow{2}{*}{$(\mathcal{O}(1))^{\log(1/\gamma)/\varepsilon}$}\\ \\
       \hline\multirow{2}{*}{Ref~\cite{wu2024efficient}} & \multirow{2}{*}{Dequantized 2D Dynamics} & \multirow{2}{*}{$\varepsilon=\mathcal O(1)$} & \multirow{2}{*}{$k$-local, 2D Symmetry} & \multirow{2}{*}{$\mathcal{O}\left(n^{e^{\log(1/\gamma)/\varepsilon}\log n}\right)$}\\ \\
      \hline\multirow{2}{*}{Theorem~\ref{thm:deq_rqite}} & \multirow{2}{*}{Dequantized RQITE} & \multirow{2}{*}{$\varepsilon>\varepsilon^*=\Omega(1)$} & \multirow{2}{*}{$k$-local} & \multirow{2}{*}{$\varepsilon^{-1}{\rm poly}\left[(\abs{S}\Delta\gamma^{-2}\varepsilon^{-1})^{\log(\Delta/(2\varepsilon^*\log(\gamma^{-2}\varepsilon^{-1})))}\right]$}\\ \\
      \hline\multirow{2}{*}{Corollary~\ref{coro:1}} & \multirow{2}{*}{Dequantized RQITE} & \multirow{2}{*}{$\varepsilon>\varepsilon^*/\|H\|$} & \multirow{2}{*}{$k$-local, $\tilde{H}=H/\|H\|$} & \multirow{2}{*}{---}\\ \\
      \hline\multirow{3}{*}{Theorem~\ref{thm:deq_rqite_2}} & \multirow{3}{*}{Dequantized RQITE} & \multirow{3}{*}{$\varepsilon=\mathcal O(1)$} & \multirow{3}{*}{\makecell[c]{$k$-local, $\gamma=1/\sqrt{2}$, \\ $\mathcal{O}(1)$-dimension}} & \multirow{3}{*}{$\mathcal{O}\left(\left(\Delta\varepsilon^{-2}{\rm poly}(\abs{S})\right)^{e^{\log(1/(\gamma\varepsilon^{2}))/\Delta}}\right)$}\\ \\ \\
      \hline\multirow{3}{*}{Corollary~\ref{coro:2}} & \multirow{3}{*}{Dequantized RQITE} & \multirow{3}{*}{$\varepsilon=\mathcal O(1/\|H\|)$} & \multirow{3}{*}{\makecell[c]{$k$-local, $\gamma=1/\sqrt{2}$, \\$\tilde{H}=H/\|H\|$,\\ $\mathcal{O}(1)$-dimension}} & \multirow{3}{*}{---}\\ \\ \\
      \hline\hline
    \end{tabular}
\end{table*}

\vspace{0.2cm}
\noindent\textbf{\emph{Dequantization of the RQITE algorithm with limited accuracy.---}} Although approximating the partition function within a constant relative error is $\# P$-hard for complex $\abs{\beta}=\mathcal{O}(1)$ in the worst-case scenario~\cite{bravyi2022quantum,goldberg2017complexity}, we demonstrate that the computational complexity may have a surprising transition when $\abs{\beta}<\beta^*$, a constant threshold related to the target Hamiltonian $H$. Specifically, we propose a dequantized algorithm for the RQITE algorithm when $\beta$ is small so that only limited accuracy can be achieved. At the heart of our dequantized algorithm is the cluster expansion for the logarithm of the partition function~\cite{wild2023classical,yin2023polynomial}, which have found utilities in approximating the partition function of high-temperature Gibbs states~\cite{yin2023polynomial,mann2024algorithmic,bakshi2024high} and Loschmidt echo~\cite{wild2023classical}. See Fig.~\ref{fig:main} (c) for an illustration of clusters; we refer to Appendix \ref{sec:cluster} for further details.
Key to the cluster expansion is that disconnected clusters will not contribute to the approximation of partition function in the form~\cite{wild2023classical}
$$\tr{\ket{x}\bra{y}e^{-\beta H}},$$
where $\ket{x}$ and $\ket{y}$ are product states. Another important assumption is that the guiding (initial) state must be classically accessible in a semi-classical form~\cite{grilo2015qma}: $\ket{\psi_c}=\sum_{j=1}^R a_j \ket{x_j}$, where $R=\mathcal{O}(\mathrm{poly}(n))$ and $\ket{x_j}$ are product states. Such an assumption will not compromise the complexity of the GLH problem as BQP-hardness holds for semi-classical guiding states~\cite{gharibian2022dequantizing}. Subsequently, our strategy is to approximate the partition function 
\begin{equation}
\begin{aligned}
    D_{\beta}(H-x)&=\bra{\psi_c}e^{-\beta (H-x)}\ket{\psi_c}\\
    &=\sum_{x,y}a_xa_y^{*}\langle y|e^{-\beta (H-x})|x\rangle
\end{aligned}
\end{equation}
by estimating each $\log(\langle y|e^{-\beta (H-x})|x\rangle)$ with the cluster expansion, which gives us the following result.
\begin{theorem}\label{thm:deq_rqite}
Suppose an $R$-configurational semi-classical guiding state $\ket{\psi_c}$ is given, and the condition in Eq.~\ref{eq:gap_assumption} holds. Then, there exists a classical algorithm to solve the GSEE problem with a runtime of
\begin{align}\label{eq:GSEE_limited}
        \frac{R^2\abs{S}}{\varepsilon}{\rm poly}\left[\left(\frac{\abs{S}}{\gamma^2\beta\varepsilon[1-\beta/\beta^*]}\right)^{\log(\beta^*/\beta)}\right],
    \end{align}
where $\beta=\Delta^{-1}\ln(\gamma^{-2}\varepsilon^{-1})$ and $\beta^*=(2e^2\mathfrak{d}(\mathfrak{d}+1))^{-1}$. The algorithm is efficient when $\beta<\beta^*$, which corresponds to the accuracy limit:
\begin{eqnarray}
    \varepsilon> \varepsilon^*=2e^2\mathfrak{d}(\mathfrak{d}+1).
\end{eqnarray}
\end{theorem}
\noindent We refer to Appendix \ref{sec:deq_limit} for the detailed proof. It is important to note that, given a constant lower bound on the accuracy $\varepsilon > \varepsilon^*$, the algorithm exhibits polynomial scaling with respect to all other parameters of the quantum system. Consequently, exponential quantum advantage is not anticipated for such problems. 

Furthermore, the accuracy scaling can be improved by normalizing the Hamiltonian to have a constant operator norm.
\begin{corollary}\label{coro:1}
The dequantization algorithm in Theorem~\ref{eq:GSEE_limited} remains efficient for a normalized Hamiltonian $\tilde{H} = H / \|H\|$ with an accuracy threshold of
$\varepsilon > \varepsilon^* / \|H\|$.
\end{corollary}

\noindent For instance, if $\|H\| = n^2$, the lower bound on the accuracy can be relaxed to $2e^2\mathfrak{d}(\mathfrak{d}+1)/n^2$. This observation underscores that the BQP-hardness of the GLH problem significantly depends on requiring an arbitrarily small inverse polynomial accuracy.


\vspace{0.2cm}
\noindent\textbf{\emph{Dequantization of the RQITE algorithm with arbitrary constant accuracy.---}}
We discuss the possibility of extending the dequantized algorithm to arbitrary constant accuracy. In general, extending the imaginary-evolution time to an arbitrary constant is still an open problem since this problem relates closely to the zero points of the (complex) partition functions of Gibbs states such that their logarithm becomes non-analytic. This problem has long been studied in (quantum) phase transition~\cite{lee1952statistical,fisher1965nature,barvinok2016combinatorics,harrow2020classical}, dating decades back. When the inverse temperature exceeds a constant threshold, the problem becomes classically intractable from the hardness of approximation statements~\cite{sly2012computational,galanis2016inapproximability,goldberg2017complexity,mann2024algorithmic}. Although these results were only found for certain quantum models, we believe their implications are general, and relate to restrictions in the imaginary-evolution time in our cases.

We circumvented the problem by determining analytic regions for the logarithm of the partition function when the guiding state overlap is large enough, which is summarized in the following result. 
\begin{lemma}\label{thm:analytic_region}
When the guiding-state overlap with the ground state satisfies $|\langle \psi_I| \psi_0\rangle|\geq \gamma= \frac{1}{\sqrt{2}},$
we have that the logarithm of the partition function $\log(D_\beta(H))$ is analytic for 
$\rm{Re}(\beta)> 0$.
\end{lemma}

\noindent The intuition behind this result is that when the ground-state component in the guiding state dominates (i.e., $\gamma^2 \geq 0.5$), the excited-state contributions diminish under the action of the imaginary-time evolution (ITE) operator, provided $\rm{Re}(\beta) > 0$. Consequently, the ground-state contribution remains unaffected by cancellations from other components, ensuring a non-vanishing partition function. We leave proof details to Theorem~\ref{thm:zero-free} in the Appendix.

The analytic region established in Lemma~\autoref{thm:analytic_region} enables the extension to arbitrary constant imaginary evolution times $\beta \in \mathbb{R} = \mathcal{O}(1)$. This approach aligns with techniques in Refs.~\cite{wild2023classical,wu2024efficient}, where analytic continuation is employed to estimate observable expectation values for constant-time quantum dynamics. Specifically, the analytic continuation is performed via the map $\beta \mapsto \beta \phi(z)$, where $\phi(z)$ transforms the disk into an elongated region encompassing the target value of $\beta$, as illustrated in Fig.~\ref{fig:main} (d). The complex function $\phi(z)$ satisfies (i) $\phi(0)=0$, (ii) $\phi(1)=1$ and (iii) $\phi(z)$ is analytic on a disk $\abs{z}\leq \nu$ with $\nu=\mathcal{O}(1)$. We provide elaborate details in the Appendix~\ref{sec:analytic_continuation}.
We assume that the semi-classical guiding state is prepared using a constant-depth quantum circuit, $\ket{\psi_c} = U\ket{0^n}$, denote the similarity-transformed Hamiltonian as $H^\prime = U^\dagger H U$, and define the complex function $f(z)=\log\left(D_{\beta\phi(z)}(H-x)\right)$. Thus, the objective becomes:
\begin{align}\label{eq:Target}
    f(1)=\log\left(D_{\beta\phi(1)}(H-x)\right)=\log\left(\langle0^n|e^{-\beta\phi(1)(H^{\prime}-x)}|0^n\rangle\right).
\end{align}
The strategy is then to (i) compute the $H^\prime=U^\dagger HU$ with the maximal degree of the corresponding interaction graph increases to $\mathfrak{d}^\prime$, (ii) approximate $f(1)$ using the cluster expansion derived by the complex Taylor approximation method [Lemma 5, Ref.~\cite{wild2023classical}]. This gives the following result:
\begin{theorem}\label{thm:deq_rqite_2}
Let the semi-classical guided state $\ket{\psi_c}=U\ket{0^n}$ that is prepared by a constant-depth quantum circuit $U$ and $\gamma=1/\sqrt{2}$. Suppose $H$ representing a $k$-local Hamiltonian defined on a $\mathcal{O}(1)$-dimensional lattice. Let the similarity-transformed Hamiltonian $H^\prime=U^\dagger HU$ and the maximum degree of its corresponding interaction graph be denoted as $\mathfrak{d}^\prime$. Then, if the condition in Eq.~\eqref{eq:gap_assumption} holds, there exists a classical algorithm that solves the GSEE problem with a run time
\begin{align}\label{eq:final_cost}
        \left[\frac{e^{2\pi\beta/\beta^*}}{\beta\varepsilon^2}{\rm{poly}}(\abs{S})\right]^{ e^{2\pi\beta/\beta^*}},
    \end{align}
where $\beta=\Delta^{-1}\ln(\gamma^{-2}\varepsilon^{-1})$, and $\beta^*=(2e^2\mathfrak{d}^\prime(\mathfrak{d}^\prime+1))^{-1}$.
\end{theorem}
\noindent The cluster expansion with the analytic continuation gives us a classical algorithm that solves the GSEE problem that scales doubly exponentially with the inverse gap and exponentially with the accuracy. 
For normalized Hamiltonian $\tilde H = H/\|H\|$, we can also improve the accuracy to $\varepsilon = \mathcal O(1/\|H\|)$.
\begin{corollary}\label{coro:2}
    The dequantization algorithm in Theorem~\ref{thm:deq_rqite_2} remains efficient for a normalized Hamiltonian $\tilde{H} = H / \|H\|$ with an accuracy $\varepsilon = \mathcal O(1/\|H\|)$.
\end{corollary}

\vspace{0.2cm}
\noindent\textbf{\emph{Discussion \& Conclusion.---}}
This work introduces an efficient dequantization algorithm for solving the guided local Hamiltonian problem. By applying cluster expansion to a randomized quantum imaginary-time evolution algorithm, we solve the ground-state energy estimation problem with limited constant accuracy. Furthermore, leveraging analytic continuation enables us to achieve arbitrary constant accuracy, provided the guiding state has a sufficiently large overlap with the ground state.
Unlike prior approaches, our results eliminate the artificial and unrealistic operator norm constraint $\|H\| \leq 1$, making them applicable to practical problems where $\|H\| = \poly(n)$. Since GLH is generally BQP-complete, this work refines the boundary between quantum and classical computational complexity, contributing to the fundamental understanding of quantum advantages.

Despite these advancements, several open questions remain. First, our algorithm achieves general constant accuracy only when the overlap $\gamma$ is large. A natural question is whether dequantizing the GSEE algorithm is possible for physical Hamiltonians with overlap $\gamma = \Omega(1/\poly(n))$, or whether classical computation faces intrinsic limitations in such cases. A promising direction might involve tensor network states~\cite{perez2007peps,anshu2024circuit,malz2024computational,wei2022sequential}, which lie at the interface of quantum and classical computing. Notably, injective PEPS (projected entangled pair states)~\cite{perez2007peps} with unique parent Hamiltonians have been shown to be quantum-efficient~\cite{schwarz2012preparing,ge2016rapid} but classically hard~\cite{anshu2024circuit} in specific settings. Further exploration of their computational complexity could offer insights and inspire new approaches to studying quantum systems and their ground states.
Another key challenge lies in the preparation of the guiding state. Our algorithm relies on the assumption of a classically accessible guiding state, which facilitates efficient classical dequantization under specific conditions. However, if the guiding state is accessible only through a quantum process, the proposed classical algorithms remain inapplicable. While numerical evidence suggests that obtaining a high-quality quantum guiding state efficiently is challenging~\cite{Lee2023NC}, this is a critical avenue for future research.

\begin{acknowledgments}
\noindent We thank Jue Xu and Runsheng Ouyang for the helpful discussion on related topics. This work is supported by the Innovation Program for Quantum Science and Technology (Grant No.~2023ZD0300200), the National Natural Science Foundation of China Grant (No.~12175003 and No.~12361161602),  NSAF (Grant No.~U2330201). 
\end{acknowledgments}

\bibliography{main}

\begin{thebibliography}{10}

\bibitem{wilson1983superconducting}
Martin~N Wilson.
\newblock Superconducting magnets.
\newblock 1983.

\bibitem{wheatley1975experimental}
John~C Wheatley.
\newblock Experimental properties of superfluid he 3.
\newblock {\em Reviews of modern physics}, 47(2):415, 1975.

\bibitem{wen1995topological}
Xiao-Gang Wen.
\newblock Topological orders and edge excitations in fractional quantum hall states.
\newblock {\em Advances in Physics}, 44(5):405--473, 1995.

\bibitem{kane2005z}
Charles~L Kane and Eugene~J Mele.
\newblock Z 2 topological order and the quantum spin hall effect.
\newblock {\em Physical review letters}, 95(14):146802, 2005.

\bibitem{kovchegov2013quantum}
Yuri~V Kovchegov and Eugene Levin.
\newblock {\em Quantum chromodynamics at high energy}.
\newblock Cambridge University Press, 2013.

\bibitem{levine2009quantum}
Ira~N Levine, Daryle~H Busch, and Harrison Shull.
\newblock {\em Quantum chemistry}, volume~6.
\newblock Pearson Prentice Hall Upper Saddle River, NJ, 2009.

\bibitem{cao2019quantum}
Yudong Cao, Jonathan Romero, Jonathan~P Olson, Matthias Degroote, Peter~D Johnson, M{\'a}ria Kieferov{\'a}, Ian~D Kivlichan, Tim Menke, Borja Peropadre, Nicolas~PD Sawaya, et~al.
\newblock Quantum chemistry in the age of quantum computing.
\newblock {\em Chemical reviews}, 119(19):10856--10915, 2019.

\bibitem{kempe2006complexity}
Julia Kempe, Alexei Kitaev, and Oded Regev.
\newblock The complexity of the local hamiltonian problem.
\newblock {\em Siam journal on computing}, 35(5):1070--1097, 2006.

\bibitem{gharibian2015quantum}
Sevag Gharibian, Yichen Huang, Zeph Landau, Seung~Woo Shin, et~al.
\newblock Quantum hamiltonian complexity.
\newblock {\em Foundations and Trends{\textregistered} in Theoretical Computer Science}, 10(3):159--282, 2015.

\bibitem{kitaev2002classical}
Alexei~Yu Kitaev, Alexander Shen, and Mikhail~N Vyalyi.
\newblock {\em Classical and quantum computation}.
\newblock Number~47. American Mathematical Soc., 2002.

\bibitem{childs2014bose}
Andrew~M Childs, David Gosset, and Zak Webb.
\newblock The bose-hubbard model is qma-complete.
\newblock In {\em Automata, Languages, and Programming: 41st International Colloquium, ICALP 2014, Copenhagen, Denmark, July 8-11, 2014, Proceedings, Part I 41}, pages 308--319. Springer, 2014.

\bibitem{o2021electronic}
Bryan O'Gorman, Sandy Irani, James Whitfield, and Bill Fefferman.
\newblock Electronic structure in a fixed basis is qma-complete.
\newblock {\em arXiv preprint arXiv:2103.08215}, 2021.

\bibitem{aharonov2002quantum}
Dorit Aharonov and Tomer Naveh.
\newblock Quantum np-a survey.
\newblock {\em arXiv preprint quant-ph/0210077}, 2002.

\bibitem{karp2010reducibility}
Richard~M Karp.
\newblock {\em Reducibility among combinatorial problems}.
\newblock Springer, 2010.

\bibitem{aharonov2013guest}
Dorit Aharonov, Itai Arad, and Thomas Vidick.
\newblock Guest column: the quantum pcp conjecture.
\newblock {\em Acm sigact news}, 44(2):47--79, 2013.

\bibitem{lin2020near}
Lin Lin and Yu~Tong.
\newblock Near-optimal ground state preparation.
\newblock {\em Quantum}, 4:372, 2020.

\bibitem{dong2022ground}
Yulong Dong, Lin Lin, and Yu~Tong.
\newblock Ground-state preparation and energy estimation on early fault-tolerant quantum computers via quantum eigenvalue transformation of unitary matrices.
\newblock {\em PRX Quantum}, 3(4):040305, 2022.

\bibitem{lin2022heisenberg}
Lin Lin and Yu~Tong.
\newblock Heisenberg-limited ground-state energy estimation for early fault-tolerant quantum computers.
\newblock {\em PRX Quantum}, 3(1):010318, 2022.

\bibitem{wan2022randomized}
Kianna Wan, Mario Berta, and Earl~T Campbell.
\newblock Randomized quantum algorithm for statistical phase estimation.
\newblock {\em Physical Review Letters}, 129(3):030503, 2022.

\bibitem{wang2023quantum}
Guoming Wang, Daniel~Stilck Fran{\c{c}}a, Ruizhe Zhang, Shuchen Zhu, and Peter~D Johnson.
\newblock Quantum algorithm for ground state energy estimation using circuit depth with exponentially improved dependence on precision.
\newblock {\em Quantum}, 7:1167, 2023.

\bibitem{ding2023even}
Zhiyan Ding and Lin Lin.
\newblock Even shorter quantum circuit for phase estimation on early fault-tolerant quantum computers with applications to ground-state energy estimation.
\newblock {\em PRX Quantum}, 4(2):020331, 2023.

\bibitem{ni2023low}
Hongkang Ni, Haoya Li, and Lexing Ying.
\newblock On low-depth algorithms for quantum phase estimation.
\newblock {\em Quantum}, 7:1165, 2023.

\bibitem{gharibian2022dequantizing}
Sevag Gharibian and Fran{\c{c}}ois Le~Gall.
\newblock Dequantizing the quantum singular value transformation: hardness and applications to quantum chemistry and the quantum pcp conjecture.
\newblock In {\em Proceedings of the 54th Annual ACM SIGACT Symposium on Theory of Computing}, pages 19--32, 2022.

\bibitem{cade2022improved}
Chris Cade, Marten Folkertsma, Sevag Gharibian, Ryu Hayakawa, Fran{\c{c}}ois~Le Gall, Tomoyuki Morimae, and Jordi Weggemans.
\newblock Improved hardness results for the guided local hamiltonian problem.
\newblock {\em arXiv preprint arXiv:2207.10250}, 2022.

\bibitem{cade2022complexity}
Chris Cade, Marten Folkertsma, and Jordi Weggemans.
\newblock Complexity of the guided local hamiltonian problem: Improved parameters and extension to excited states.
\newblock {\em arXiv preprint arXiv:2207.10097}, 2022.

\bibitem{gall2024classical}
Fran{\c{c}}ois~Le Gall.
\newblock Classical algorithms for constant approximation of the ground state energy of local hamiltonians.
\newblock {\em arXiv preprint arXiv:2410.21833}, 2024.

\bibitem{gilyen2019quantum}
Andr{\'a}s Gily{\'e}n, Yuan Su, Guang~Hao Low, and Nathan Wiebe.
\newblock Quantum singular value transformation and beyond: exponential improvements for quantum matrix arithmetics.
\newblock In {\em Proceedings of the 51st Annual ACM SIGACT Symposium on Theory of Computing}, pages 193--204, 2019.

\bibitem{tang2022dequantizing}
Ewin Tang.
\newblock Dequantizing algorithms to understand quantum advantage in machine learning.
\newblock {\em Nature Reviews Physics}, 4(11):692--693, 2022.

\bibitem{wu2024efficient}
Yusen Wu, Yukun Zhang, and Xiao Yuan.
\newblock An efficient classical algorithm for simulating short time 2d quantum dynamics.
\newblock {\em arXiv preprint arXiv:2409.04161}, 2024.

\bibitem{wild2023classical}
Dominik~S Wild and {\'A}lvaro~M Alhambra.
\newblock Classical simulation of short-time quantum dynamics.
\newblock {\em PRX Quantum}, 4(2):020340, 2023.

\bibitem{yin2023polynomial}
Chao Yin and Andrew Lucas.
\newblock Polynomial-time classical sampling of high-temperature quantum gibbs states.
\newblock {\em arXiv preprint arXiv:2305.18514}, 2023.

\bibitem{mann2024algorithmic}
Ryan~L Mann and Romy~M Minko.
\newblock Algorithmic cluster expansions for quantum problems.
\newblock {\em PRX Quantum}, 5(1):010305, 2024.

\bibitem{bakshi2024high}
Ainesh Bakshi, Allen Liu, Ankur Moitra, and Ewin Tang.
\newblock High-temperature gibbs states are unentangled and efficiently preparable.
\newblock {\em arXiv preprint arXiv:2403.16850}, 2024.

\bibitem{sly2012computational}
Allan Sly and Nike Sun.
\newblock The computational hardness of counting in two-spin models on d-regular graphs.
\newblock In {\em 2012 IEEE 53rd Annual Symposium on Foundations of Computer Science}, pages 361--369. IEEE, 2012.

\bibitem{goldberg2017complexity}
Leslie~Ann Goldberg and Heng Guo.
\newblock The complexity of approximating complex-valued ising and tutte partition functions.
\newblock {\em computational complexity}, 26:765--833, 2017.

\bibitem{haah2024learning}
Jeongwan Haah, Robin Kothari, and Ewin Tang.
\newblock Learning quantum hamiltonians from high-temperature gibbs states and real-time evolutions.
\newblock {\em Nature Physics}, pages 1--5, 2024.

\bibitem{bravyi2022quantum}
Sergey Bravyi, Anirban Chowdhury, David Gosset, and Pawel Wocjan.
\newblock Quantum hamiltonian complexity in thermal equilibrium.
\newblock {\em Nature Physics}, 18(11):1367--1370, 2022.

\bibitem{grilo2015qma}
Alex~Bredariol Grilo, Iordanis Kerenidis, and Jamie Sikora.
\newblock Qma with subset state witnesses.
\newblock In {\em International Symposium on Mathematical Foundations of Computer Science}, pages 163--174. Springer, 2015.

\bibitem{lee1952statistical}
Tsung-Dao Lee and Chen-Ning Yang.
\newblock Statistical theory of equations of state and phase transitions. ii. lattice gas and ising model.
\newblock {\em Physical Review}, 87(3):410, 1952.

\bibitem{fisher1965nature}
M.E. Fisher.
\newblock {\em The Nature of Critical Points}.
\newblock University of Colorado Press, 1965.

\bibitem{barvinok2016combinatorics}
Alexander Barvinok.
\newblock {\em Combinatorics and complexity of partition functions}, volume~30.
\newblock Springer, 2016.

\bibitem{harrow2020classical}
Aram~W Harrow, Saeed Mehraban, and Mehdi Soleimanifar.
\newblock Classical algorithms, correlation decay, and complex zeros of partition functions of quantum many-body systems.
\newblock In {\em Proceedings of the 52nd Annual ACM SIGACT Symposium on Theory of Computing}, pages 378--386, 2020.

\bibitem{galanis2016inapproximability}
Andreas Galanis, Daniel {\v{S}}tefankovi{\v{c}}, and Eric Vigoda.
\newblock Inapproximability of the partition function for the antiferromagnetic ising and hard-core models.
\newblock {\em Combinatorics, Probability and Computing}, 25(4):500--559, 2016.

\bibitem{perez2007peps}
David Perez-Garcia, Frank Verstraete, J~Ignacio Cirac, and Michael~M Wolf.
\newblock Peps as unique ground states of local hamiltonians.
\newblock {\em arXiv preprint arXiv:0707.2260}, 2007.

\bibitem{anshu2024circuit}
Anurag Anshu, Nikolas~P Breuckmann, and Quynh~T Nguyen.
\newblock Circuit-to-hamiltonian from tensor networks and fault tolerance.
\newblock In {\em Proceedings of the 56th Annual ACM Symposium on Theory of Computing}, pages 585--595, 2024.

\bibitem{malz2024computational}
Daniel Malz and Rahul Trivedi.
\newblock Computational complexity of isometric tensor network states.
\newblock {\em arXiv preprint arXiv:2402.07975}, 2024.

\bibitem{wei2022sequential}
Zhi-Yuan Wei, Daniel Malz, and J~Ignacio Cirac.
\newblock Sequential generation of projected entangled-pair states.
\newblock {\em Physical Review Letters}, 128(1):010607, 2022.

\bibitem{schwarz2012preparing}
Martin Schwarz, Kristan Temme, and Frank Verstraete.
\newblock Preparing projected entangled pair states on a quantum computer.
\newblock {\em Physical review letters}, 108(11):110502, 2012.

\bibitem{ge2016rapid}
Yimin Ge, Andr{\'a}s Moln{\'a}r, and J~Ignacio Cirac.
\newblock Rapid adiabatic preparation of injective projected entangled pair states and gibbs states.
\newblock {\em Physical review letters}, 116(8):080503, 2016.

\bibitem{Lee2023NC}
Seunghoon Lee, Joonho Lee, Huanchen Zhai, Yu~Tong, Alexander~M. Dalzell, Ashutosh Kumar, Phillip Helms, Johnnie Gray, Zhi-Hao Cui, Wenyuan Liu, Michael Kastoryano, Ryan Babbush, John Preskill, David~R. Reichman, Earl~T. Campbell, Edward~F. Valeev, Lin Lin, and Garnet Kin-Lic Chan.
\newblock Evaluating the evidence for exponential quantum advantage in ground-state quantum chemistry.
\newblock {\em Nature Communications}, 14(1):1952, Apr 2023.

\bibitem{tang2019quantum}
Ewin Tang.
\newblock A quantum-inspired classical algorithm for recommendation systems.
\newblock In {\em Proceedings of the 51st annual ACM SIGACT symposium on theory of computing}, pages 217--228, 2019.

\bibitem{bravyi2021classical}
Sergey Bravyi, David Gosset, and Ramis Movassagh.
\newblock Classical algorithms for quantum mean values.
\newblock {\em Nature Physics}, 17(3):337--341, 2021.

\bibitem{giovannetti2004quantum}
Vittorio Giovannetti, Seth Lloyd, and Lorenzo Maccone.
\newblock Quantum-enhanced measurements: beating the standard quantum limit.
\newblock {\em Science}, 306(5700):1330--1336, 2004.

\bibitem{an2023linear}
Dong An, Jin-Peng Liu, and Lin Lin.
\newblock Linear combination of hamiltonian simulation for nonunitary dynamics with optimal state preparation cost.
\newblock {\em Physical Review Letters}, 131(15):150603, 2023.

\bibitem{huo2023error}
Mingxia Huo and Ying Li.
\newblock Error-resilient monte carlo quantum simulation of imaginary time.
\newblock {\em Quantum}, 7:916, 2023.

\bibitem{bjorklund2008fast}
Andreas Bj{\"o}rklund, Thore Husfeldt, Petteri Kaski, and Mikko Koivisto.
\newblock The fast intersection transform with applications to counting paths.
\newblock {\em arXiv preprint arXiv:0809.2489}, 2008.

\bibitem{ansari2005more}
Arsalan~Hojjat Ansari, Mohammad~Sal Moslehian, et~al.
\newblock More on reverse triangle inequality in inner product spaces.
\newblock {\em International journal of mathematics and mathematical sciences}, 2005:2883--2893, 2005.

\end{thebibliography}
\clearpage
\widetext
\appendix

\section{Comparison with existing results}
This section discusses the dequantization of algorithms that solve the ground-state energy estimation (GSEE) problem. We will use the definition and notation defined in the main text for the GSEE problem.

In recent years, rapid development in algorithm design has been witnessed in solving the GSEE problem. On the one hand, quantum algorithms~\cite{lin2020near,dong2022ground} with near-optimal query complexity performance have been found using the quantum singular value transformation (QSVT)~\cite{gilyen2019quantum} or the quantum eigenvalue transformation algorithms. The algorithm can be viewed as a way of performing polynomial transformations of the Hamiltonian so that one can approximate the target function with a polynomial function. Quantum algorithms tailored for early fault-tolerance quantum devices~\cite{lin2022heisenberg,wan2022randomized,wang2023quantum,ding2023even,ni2023low} thrive intending to reduce the query depth. On the other hand, theoretical understandings have depended. It is long known that the local Hamiltonian problem is QMA-complete~\cite{kitaev2002classical,kempe2006complexity}, indicating even quantum computers cannot efficiently solve the problem unless BQP=QMA. Yet, it is found~\cite{gharibian2022dequantizing} that when a guiding state $\ket{\phi_I}$ that has a non-trivial (i.e., at least polynomially small) overlap with the ground state is accessible, such problems become BQP-hard. It should be noted that a lower-bounded estimation of the overlap $\gamma$ (see the main text), which satisfies $|\langle \psi_I |\psi_0\rangle|\geq \gamma$, is assumed to be known. The first dequantization algorithm is proposed by Ref.~\cite{gharibian2022dequantizing} such that when the guiding state is classically accessible through the sample and query (SQ) model~\cite{tang2019quantum,tang2022dequantizing}, we can dequantize the quantum algorithm based on QSVT for solving the GSEE problem as long as the accuracy $\epsilon$ stays a constant. Unfortunately, the dequantization algorithms provided in Ref.~\cite{gharibian2022dequantizing} only work when the Hamiltonian satisfies $\|H\|\leq 1$, a condition violated by most physical Hamiltonians such that $\|H\|=\mathcal{O}(\text{poly}(n))$. To see why this is the case, we briefly review the dequantization idea and mildly adapt the original singular value transformation to the eigenvalue transformation that better suits our purpose. Suppose that we are given SQ access to a guiding state $u(x)=\sum_i u_i\ket{i}$. That is one can sample $i$ according to the probability $p_i:=|u_i|^2$ and query the values $u_i$ for arbitrary $i$. When the Hamiltonian $H$ is $s$-sparse, we can then efficiently estimate the following quantity within desirable accuracy $\bra{u} P(H) \ket{u}$, where $P(x)$ is a polynomial function of constant degree $d$. At each iteration, the dequantization algorithm works by first sampling a component $i$ from the guiding state and then computing the estimator $z_i:= \bra{i} P(H) \ket{u}/u_i.$ The key to the dequantization idea lies in the sparsity of the Hamiltonian. That is to compute $\bra{i} P(H) \ket{u}$, we only need to compute $s^d$ entries of $P(H)\ket{u}$, which takes $O(s^{2d})$ steps. The final estimation to $\bra{u} P(H) \ket{u}$ is given by the empirical mean of the estimator $\bar{z}:=\sum_i^M z_i$, where $M$ is the number of samples. Conceptually, the formulation can be seen as a classical evaluation method for the eigenvalue transformation of a sparse matrix (e.g., some Hamiltonian) when the transformation consists of polynomial functions of constant degree. For the GSEE problem, the shifted sign function along with a polynomial function of degree $\widetilde{\mathcal{O}}(\epsilon^{-1})$ to approximate the function is found in the QSVT algorithm to provide near-optimal performance~\cite{lin2020near}. For physical Hamiltonians, the truncation degree of the polynomial function that approximates the shifted sign function will grow up to $\widetilde{\mathcal{O}}(\|H\|\epsilon^{-1})$, rendering the dequantization scheme intractable.

Our dequantization algorithm on the other hand focuses on early fault-tolerant quantum algorithms such as the randomized Fourier estimation (RFE) methods~\cite{lin2022heisenberg,wan2022randomized,wang2023quantum}, which we will discuss in the next section. To this end, we first propose the randomized quantum imaginary-time evolution (RQITE) algorithm, where the primary task is to estimate the partition function. The form of the partition function (exponential of the Hamiltonian resulting from the imaginary-time evolution) then admits the approximation of cluster expansion~\cite{haah2024learning,wild2023classical,mann2024algorithmic,yin2023polynomial,bakshi2024high}, an approach that has been applied to study classical simulation of Hamiltonian dynamics and Gibbs state. We assume that the guiding state is a semi-classical state~\cite{grilo2015qma}, which is strictly weaker than the SQ model. This gives the result that for accuracy no lesser than $2e^2\mathfrak{d}(\mathfrak{d}+1)$, where $\mathfrak{d}$ is the degree of the Hamiltonian, we achieve classical algorithms with polynomial dependence on $\gamma$, $\mathfrak{d}$, $\epsilon$ and superpolynomial on the spectral gap $\Delta$. In general, the result cannot be improved to arbitrary constant accuracy due to the hardness of approximation results from the study of Gibbs states. Yet, we find that when the guiding state has a large enough (constant) overlap with the ground state, we are allowed to apply analytic continuation to accomplish arbitrary constant accuracy of the ground-state energy. The analytic continuation is first proposed in Ref.~\cite{wild2023classical} for realizing expectation value estimation of local observable for Hamiltonian evolution of arbitrary constant time. The result has later been improved~\cite{wu2024efficient} to estimate global observables for 2D quantum systems using techniques from shallow quantum circuit simulation~\cite{bravyi2021classical}. The cluster expansion in combination with analytic continuation gives us classical algorithms that have double exponential dependence on $\Delta$, $\mathfrak{d}$, and superpolynomial on $\epsilon$. Ref.~\cite{wu2024efficient} also discusses the application of dequantization of RFE algorithms. The performance of its classical algorithm is similar to this work. The key technique there is the utility of the ancillary-free Hadamard test. Yet, the construction is only available for 2D geometrically local Hamiltonians, whereas the results in this work can be applied to Hamiltonians with arbitrary geometry as long as the interaction is sparse. This makes our methods suitable for quantum systems ranging from quantum spin to many-body, and molecular systems.

Finally, we comment on the relationship between our dequantized algorithm with recent progress regarding the dequantization of algorithms for the guided local Hamiltonian (GLH) problems. In Ref.~\cite{gharibian2022dequantizing}, the authors proved that the GLH problem is BQP-hard for Hamiltonian $\|H\|\leq1$ by encoding a BQP machine (i.e.~quantum circuit) into the GLH problem using the Feynman-Kiteav clock construction~\cite{kitaev2002classical,kempe2006complexity}. That is our classical algorithms achieve $\mathcal{O}(1)$-accuracy estimation of Hamiltonian with $\|H\|=\mathcal{O}(\rm{poly}(n))$ may imply inverse-polynomial-accuracy of Hamiltonian with $\|H\|=\mathcal{O}(1)$. Here, we remark that the constraint on the operator norm of the Hamiltonian does not contradict our result. But it is rather more natural to think general physical Hamiltonian that $\|H\|=\mathcal{O}(\rm{poly}(n))$ as suggested by the Feynman-Kiteav clock construction~\cite{kitaev2002classical,kempe2006complexity}, whose hardness result is now only proven for inverse-polynomial accuracy. In fact, the author in Ref.~\cite{gharibian2022dequantizing} achieves such constraint by manually dividing the Hamiltonian by its operator norm.

\section{Recent developments of ground-state energy estimation problems}
\label{sec:gsee_dev}
Recent progress has witnessed great advancement in the GSEE problem, especially algorithms tailored for early-fault tolerant quantum devices. Here, we briefly review the recent developments of such algorithms in two aspects: random Fourier estimation (RFE)~\cite{lin2022heisenberg,wan2022randomized,wang2023quantum} and eigenvalue estimation based on classical signal processing techniques~\cite{ding2023even,ni2023low}.

We now provide a formal description of the GSEE problem and corresponding assumptions:
\begin{definition}[Assumptions of the initial state and energy gap, and the ground state energy estimation problem]
\label{def:gsee}
Given a $n$-qubit Hamiltonian $\hat H$, let $E_0<E_1\leq\cdots\leq E_{N-1}$, where $N=2^n$, be its eigenvalues with corresponding eigenstates $\ket{\psi_0},\ket{\psi_1},\cdots,\ket{\psi_{N-1}}$. Suppose a lower-bounded estimation $\Delta$ for the energy gap is given such that $E_1-E_0\geq \Delta$. Suppose we can also prepare an initial state $\ket{\psi_I}$ such that its overlap with the ground state satisfies $|\langle \psi_I| \psi_0\rangle|\geq \gamma$. The GSEE task is to obtain an estimate of $E_0^\prime$ for the ground-state energy such that $|E_0-E_0^\prime|\leq\epsilon$.
\end{definition}
These problems are known as local Hamiltonian (LH) problems in the study of Hamiltonian complexity~\cite{kitaev2002classical,kempe2006complexity,gharibian2015quantum}, where QMA-complete results are proven for both artificial~\cite{kitaev2002classical,kempe2006complexity} and physical quantum systems~\cite{childs2014bose,o2021electronic}. Yet, the near-optimal quantum algorithm for solving the GSEE problem gives the query complexity $\widetilde{\mathcal{O}}(\gamma^{-1}\epsilon^{-1})$, indicating that when a reasonable initial state (i.e.~has greater than exponentially small overlap with the ground state) is available, the problem becomes efficiently solvable. This point of view is theoretically verified by Ref.~\cite{gharibian2022dequantizing}, where the guided local Hamiltonian (GLH) problem is formalized to characterize LH problems with a given guiding (initial) state. The GLH is proven to be BQP-hard and a dequantization scheme is discussed as we mentioned in the last section.

While great progress has been made in solving the GSEE problem, we mainly focus on the development of methods based on the Hadamard test realization. To this end,  Ref.~\cite{lin2022heisenberg} gives exemplary results that achieve the Heisenberg-limited performance~\cite{giovannetti2004quantum}, which is known to be optimal for the dependence on the accuracy. We elaborate on their idea here with some small modifications to facilitate later discussion. First, let us define the spectrum function $P(x)$ of the initial state and the convolution function
\begin{equation}\label{eq:intital_state_spectrum}
\begin{aligned}
    P(x):=&\sum_{j=0}^{N-1} p_j\delta(x-E_j),\\
    C(x):=&(f*P)(x)=\sum_{j=0}^{N-1} p_j\cdot f(x-E_j),
\end{aligned}
\end{equation}
where $p_j=|\langle \psi_I|\psi_j\rangle |^2$, and $f(x)$ is a chosen function that we will specify later. The key idea of Ref.~\cite{lin2022heisenberg} is by analyzing the convolution function with a properly chosen function $f(x)$, one can pinpoint the ground-state energy with the knowledge based on the initial-state overlap, i.e.~$\gamma$ in Definition \ref{def:gsee}. Specifically, Ref.~\cite{lin2022heisenberg} applies a $2\pi$-periodic Heaviside function $(f(x):=0,\textrm{if}~x<0;1,\textrm{otherwise})$, which makes the $C(x)$ the cumulative distribution function (CDF): $C(x)=\sum_{j:E_j\leq x}p_j$. As such, one could binary search on a promised interval $[E_a,E_b]$ that contains $E_0$. The question then left to estimate the convolution function by a quantum algorithm. To this end, the convolution theorem is used: The Fourier transformation of a convolution function equals the multiplication of the Fourier transformations of the two functions involved in the convolution. As such, we have
\begin{equation}\label{eq:c_ft}
\begin{aligned}
    C(x)&=\mathcal{F}^{-1}\left(\hat{f}\cdot \hat{P}\right)(x)\\
    &=\int_{-\infty}^{\infty} \hat{f}(t) e^{i x t} \operatorname{Tr}\left(\rho_I e^{-i H t}\right) \mathrm{d} t\equiv \mathbb{E}_{t\sim\hat{f}(t)} [ e^{i x t}X(t)],
\end{aligned}
\end{equation}
where $\hat{f}(t)$ and $\hat{P}(t)$ are the Fourier-transformed function of $f(x)$ and $P(x)$ respectively, $\rho_I:=\ket{\psi_I}\bra{\psi_I}$, and $X(t)=\operatorname{Tr}\left(\rho_I e^{-i H t}\right)$. 
Henceforth, we may refer to $f(x)$ as the filter function.
For $\hat{f}(t)$ that measures non-negative over $\mathbb{R}$, one can view it as a probabilistic distribution over $t$ up to some normalization.
Otherwise, we may extract the phase and sample according to the modulus of the function and add the phase back to the samples.
Furthermore, a truncation of the integral in Eq.~\eqref{eq:c_ft} to $[-T,T]$ may be applied. Thus, in the last line of Eq.~\eqref{eq:c_ft}, we have viewed $X(t)$ as a random variable. The algorithm then follows by generating a set of $\{t\}$ from $\hat{f}(t)$ and estimating each $X(t)$ through the Hamdamard test circuit with the controlled unitary to be a Hamiltonian evolution operation. Subsequently, we estimate $C(x)$ for different $x$ by multiplying the corresponding phase given by Eq.~\eqref{eq:c_ft} and taking the empirical mean for approximating the expectation value. To detect the ground-state energy, we need to distinguish whether the CDF is close enough to zero or at least $\gamma^2$. Hence, $\gamma^2$-accuracy of estimating the CDF is required, yielding a $\widetilde{\mathcal{O}}(\gamma^{-4})$ sample complexity according to the Chernoff bound. The maximal evolution time is given by $t_{\text{max}}=T=\widetilde{\mathcal{O}}(\epsilon^{-1})$. This gives a total evolution time $t_{\text{total}}=T=\widetilde{\mathcal{O}}(\epsilon^{-1}\gamma^{-4})$. 

Significant progress has been made in developing the RFE methods with the central variants to be the choice of the function $f(x)$ in the convolution function. To this end, we resort to the exponential function. Interestingly, we show the connection between the choice of filter function and the imaginary-time evolution (ITE) methods. Under the assumption that energy gap $\Delta$ is greater than the accuracy, our methods can efficiently solve the GSEE problem. Perhaps the most relevant work to ours is Ref.~\cite{wang2023quantum}, where the Gaussian derivative function $-\frac{1}{\sqrt{2 \pi} \sigma^3} x e^{-\frac{x^2}{2 \sigma^2}}$ is considered. Applying such a function, the determination of the ground-state energy then becomes an evaluation of whether the convolution function goes to zero, which is similar to our construction, which will be discussed in the next section. As such, the maximal and total evolution time are $\mathcal{\widetilde{O}}(\Delta^{-1})$ and $\widetilde{\mathcal{O}}\left(\epsilon^{-2} \gamma^{-4} \Delta\right)$ for our and Ref.~\cite{wang2023quantum}.

Interestingly, when the initial state overlap is large enough (i.e.~beyond some constant threshold), one can further reduce the query depth: as the overlap becomes larger, the Hamiltonian evolution time in one coherent run becomes smaller. The first such algorithm is proposed by Ref.~\cite{ding2023even}, where the authors cast the eigenvalue estimation to an optimization problem from signals detected from the Hadmard test circuit. The overlap threshold is given by $\gamma^2=0.71$. This result is later improved~\cite{ni2023low} by means of robust phase estimation to $\gamma^2=4-2\sqrt{3}$.

\section{GSEE algorithm based on imaginary-time evolution}
\subsection{The randomized quantum imaginary-time evolution algorithm}
\label{sec:rqite}
In this section, we propose a nascent algorithm, dubbed randomized quantum imaginary-time evolution (RQITE), with the filter (projection) function in Eq.~\eqref{eq:c_ft} given by the exponential function.

Let us first define the filter operator $\Upsilon_\beta(x)$ as
\begin{equation}\label{eq:filter_function}
    \Upsilon_\beta(x):=e^{\beta x}.
\end{equation}
We also define the partition function w.r.t. the initial state as
\begin{equation}\label{eq:partition_function}
    D_\beta(H-x):=\bra{\psi_I} e^{-\beta (H-x)} \ket{\psi_I}.
\end{equation}
Then, we have the following theorem which states the equivalence between the partition function $D_\beta(x)$ and convolution function given by Eq.~\eqref{eq:c_ft} with $f(x)$ being substituted by the projection function $\Upsilon_\beta(x)$. This equivalence is key to our construction of the dequantization of the quantum algorithms. That is to evaluate the complicated convolution function, we can equivalently estimate the partition function, which is more amenable to the cluster expansion.
\begin{theorem}[Equivalence between $C(x)$ and $D_\beta(H-x)$]
\label{theorem:partition_convolution_equivalence}
Let $D_\beta(H-x)$ be defined by Eq.~\eqref{eq:partition_function}. Let also $C(x)$ be the convolution of the initial state spectrum with projection function defined by Eq.~\eqref{eq:filter_function}. Then, we have
\begin{equation}
    D_\beta(H-x)\equiv C(x).
\end{equation}
\end{theorem}
\begin{proof}
The equivalence between $C(x)$ and $D_\beta(H-x)$ can be shown by  Eq.~\eqref{eq:partition_function} in the eigenbasis of the Hamiltonian,
\begin{equation}\label{eq:convolution_partition_equiv}
\begin{aligned}
    D_\beta(H-x)&=\bra{\psi_I} e^{-\beta (H-x)} \ket{\psi_I}\\
    &=\sum_{j=0}^{N-1} |\langle\psi_I|\psi_j\rangle |^2 \cdot e^{\beta (x-E_j)}\\
    &=\sum_{j=0}^{N-1} p_j \cdot \Upsilon_\beta(x-E_j) \equiv C(x).
\end{aligned}
\end{equation}
\end{proof}

By establishing the above equivalence, we can analyze the convoluted spectrum through the partition function $D_\beta(H-x)$.

Given the lower-bounded $E_a$ and upper-bounded estimation $E_b$ to the ground state energy $E_0$, we provide the following strategy for estimating $E_0$ using the partition function $D_\beta(H-x)$ of the shifted ITE operator $\Lambda_\beta(H-x)$. First, let us consider two different partition functions $D_\beta(H-x)$ and $D_{2\beta}(H-x)$. We define the difference between the two partition functions between the interval $[E_a,E_b]$:
\begin{equation}\label{eq:partition_difference}
    R(x) := D_\beta(H-x) - D_{2\beta}(H-x),x\in [E_a,E_b].
\end{equation}
We plot a sketch of $D_\beta(H-x)$, $D_{2\beta}(H-x)$ and $R(x)$ in Fig.~\ref{fig:main} (a)  in the main text. Starting from $x=E_a$, the value of $R(x)$ first rises until it reaches the point that is denoted as $E_\mathrm{max}$. Then, the value of $R(x)$ falls and approaches zero as  $x$ becomes closer to $E_0$. Therefore, we can shrink the interval that we find $E_0$ down to $[E_\textrm{max},E_b]$. Note that there are cases in which the value of $R(x)$ falls directly at the starting point $E_a$, e.g., $E_\mathrm{max}=E_a$, which is the case in our illustrative graph.

The above observation prompts us to evaluate $R(x)$ at each $\epsilon$ interval starting from $E_a$ for finding $E_0$. Intuitively, if the value of $R(x)$ starts to go down and we keep estimating $R(x)$ until its value bellows a certain threshold, then, we find the proper estimation $E_0^\prime$ for $E_0$ that $|E_0^\prime-E_0|\leq\epsilon$. Formally, the following lemma specifies the value of $\beta$ and the threshold that we use to determine if the approximation for the ground state energy is found, and the RQITE algorithm can terminate.

Furthermore, as pointed out in Ref.~\cite{wang2023quantum}, the energy gap of quantum molecular systems is often greater than the accuracy $\epsilon$ by orders. Taking advantage of this property, Ref.~\cite{wang2023quantum} is able to suppress the maximal evolution time from $\mathcal{O}(\epsilon^{-1})$ to $\mathcal{O}(\Delta^{-1})$ for the GSEE problems. To exploit this property, we make the following assumption.
\begin{assumption}[Relation between the lower bound of the energy gap and the accuracy]
\label{assume:gap_accuracy_relation}
Let $\Delta$ and $\gamma$ be the quantities that are defined in Definition~\ref{def:gsee}, and $\epsilon$ is the accuracy that we want to achieve for estimating the ground state energy. Then, it is assumed that the gap satisfies
\begin{equation}
    \frac{\Delta}{\epsilon}\geq\ln(\gamma^{-2}\epsilon^{-1}).
\end{equation}
\end{assumption}
\noindent In the following, let us provide justification and evidence for Assumption~\ref{assume:gap_accuracy_relation}.

Exploiting the above-proposed assumption, we provide the following lemma to specify the threshold that the RQITE algorithm can stop.
\begin{lemma}[Value of $\beta$ and termination threshold of the RQITE algorithm for the GSEE problems]
\label{lemma:beta_value}
Let $\gamma$ and $\Delta$ be the quantities that are given in Definition~\ref{def:gsee}. Let $D_\beta(H-x)$ and $D_{2\beta}(H-x)$ be defined by Eq.~\eqref{eq:partition_function}, and their difference $R(x)$ is defined by Eq.~\eqref{eq:partition_difference}. Let $\Xi$ be the termination threshold for the RQITE algorithm. Then, for
\begin{equation}\label{eq:beta_value}
    \beta:=\frac{\ln(\gamma^{-2}\epsilon^{-1})}{\Delta},
\end{equation}
we have
\begin{equation}\label{eq:r_value}
\begin{aligned}
    R(x)&<0.5p_0\beta\epsilon+p_0\epsilon,&\forall  x\in[E_0-\frac{\epsilon}{2}, E_0];\\
    R(x)&>p_0\beta\epsilon,&\forall x\in[E_{\mathrm{max}}, E_0-\epsilon].
\end{aligned}
\end{equation}
Then, the termination threshold is given by
\begin{equation}\label{eq:termination_threshold}
    \Xi:=\left(\frac{\beta}{2}+1\right)p_0\epsilon.
\end{equation}
\end{lemma}
\begin{proof}
First, let us explain why we care about $x$ in the interval $[E_0-\frac{\epsilon}{2}, E_0]$ but not $[E_0-\frac{\epsilon}{2}, E_0+\frac{\epsilon}{2}]$. This is because as soon as $x$ reaches the interval $[E_0, E_0+\frac{\epsilon}{2}]$, we either get $R(x)\approx 0$ or $R(x)\leq 0$ from which we can directly determine that the estimation for the ground state energy is found.
Therefore, we only need to discriminate $R(x)$ for $x\in[E_0-\frac{\epsilon}{2}, E_0]$ from $x\in[E_{\mathrm{max}}, E_0-\epsilon]$.

Next, to prove this lemma, we take advantage of the equivalence between the partition function that we proved in Theorem~\ref{theorem:partition_convolution_equivalence}. In the following, we discuss the proof in the language of the spectrum function of the initial state that is defined in Eq.~\eqref{eq:intital_state_spectrum}.

First, let us consider $x$ in the interval $[E_0-\frac{\epsilon}{2}, E_0]$. We have
\begin{equation}\label{eq:beta_proof_1}
\begin{aligned}
    R(x)&=p_0\left( \Upsilon_\beta(x-E_0)-\Upsilon_{2\beta}(x-E_0) \right)+\sum_{j=1}^{N-1}p_j \left(\Upsilon_\beta(x-E_j)-\Upsilon_{2\beta}(x-E_j) \right)\\
    &\leq p_0\left( \Upsilon_\beta(x-E_0)-\Upsilon_{2\beta}(x-E_0) \right)+\max_{j\in[1,N-1]}  \left(\Upsilon_\beta(x-E_j)-\Upsilon_{2\beta}(x-E_j) \right).
\end{aligned}
\end{equation}
We note that for $x\in[E_0-\frac{\epsilon}{2}, E_0]$, the first term dominants most part of $R(x)$ and has the same tendency as $x$ varies that it decays as $x$ becomes larger in this interval. Hence, the first term takes the maximum value at point $x=E_0-\frac{\epsilon}{2}$, and we get
\begin{equation}\label{eq:beta_proof_2}
\begin{aligned}
    p_0\left( \Upsilon_\beta(x-E_0)-\Upsilon_{2\beta}(x-E_0) \right) &\leq p_0 \left(\Upsilon_\beta\left(-\frac{\epsilon}{2}\right) - \Upsilon_{2\beta}\left(-\frac{\epsilon}{2}\right)\right)\\
    &=p_0(e^{-\frac{\beta\epsilon}{2}}-e^{-\beta\epsilon})\\
    &=0.5p_0\beta\epsilon+p_0\sum_{k=2}^\infty \frac{1}{k!}\left(\left(\frac{-\beta\epsilon}{2}\right)^k-(-\beta\epsilon)^k \right)\\
    &<0.5p_0\beta\epsilon,
\end{aligned}
\end{equation}
where in the third line we have used the Taylor expansion. The last line of Eq.~\eqref{eq:beta_proof_2} follows from Assumption~\ref{assume:gap_accuracy_relation} that $\beta\epsilon=\frac{\ln(\gamma^{-2}\epsilon^{-1})\epsilon}{\Delta}\leq1$. Note that the second term in the third line of Eq.~\eqref{eq:beta_proof_2} is a summation of negative and positive terms, alternatively. That is $\frac{1}{2!}\cdot\frac{-3(\beta\epsilon)^2}{4}+\frac{1}{3!}\cdot\frac{7(\beta\epsilon)^3}{8}+\cdots$. Besides, the negative terms are larger than the adjacent positive terms in magnitude because $\beta\epsilon\leq1$. Therefore, the summation adds up to a negative factor.

For the second term in the last line of Eq.~\eqref{eq:beta_proof_1}, let us first note that
\begin{equation}\label{eq:beta_proof_3}
\begin{aligned}
    x-E_j&\leq E_0-E_j\\
    &\leq -\Delta.
\end{aligned}
\end{equation}
Note that for the two terms in the second term, $\Upsilon_{\beta}(x-E_j)$ and $\Upsilon_{2\beta}(x-E_j)$, their values are exponentially suppressed owing to $x-E_j<0$ for $x\in[E_0-\frac{\epsilon}{2}, E_0]$. Because $\Upsilon_{\beta}(x-E_j)$ decays slower than $\Upsilon_{2\beta}(x-E_j)$, the maximum value of the second term is taken when $x-E_j$ is maximized. Therefore, we have
\begin{equation}\label{eq:beta_proof_4}
\begin{aligned}
    \max_{j\in[1,N-1]} \left(\Upsilon_\beta(x-E_j)-\Upsilon_{2\beta}(x-E_j) \right) &\leq \Upsilon_\beta(-\Delta)-\Upsilon_{2\beta}(-\Delta) \\
    &= e^{-\beta\Delta}-e^{-2\beta\Delta}\\
    &=\gamma^2\epsilon-(\gamma^2\epsilon)^2\\
    &< p_0 \epsilon,
\end{aligned}
\end{equation}
where in the first line we have relaxed $p_j$ to $1$; we plug in Eq.~\eqref{eq:beta_value} in the third line; and in the last line, we have used the fact that $\gamma$ serves as a lower-bounded estimation for $p_0$.
This completes the proof for the first line of Eq.~\eqref{eq:r_value}. The maximal value of $R(x)$ in this interval also decides the termination threshold $\Xi$.

Next, for $x\in[E_{\mathrm{max}}, E_0-\epsilon]$, we have
\begin{equation}\label{eq:beta_proof_5}
\begin{aligned}
    R(x)&=p_0\left( \Upsilon_\beta(x-E_0)-\Upsilon_{2\beta}(x-E_0) \right)+\sum_{j=1}^{N-1}p_j \left(\Upsilon_\beta(x-E_j)-\Upsilon_{2\beta}(x-E_j) \right)\\
    &\geq p_0\left( \Upsilon_\beta(x-E_0)-\Upsilon_{2\beta}(x-E_0) \right)+\min_{j\in[1,N-1]} p_j \left(\Upsilon_\beta(x-E_j)-\Upsilon_{2\beta}(x-E_j) \right)\\
    &> p_0\left( \Upsilon_\beta(-\epsilon)-\Upsilon_{2\beta}(-\epsilon) \right)\\
    &=p_0(e^{-\beta\epsilon}-e^{-2\beta\epsilon}),
\end{aligned}
\end{equation}
where in the third line we let $x=E_0-\epsilon$ because the first term decays as $x-E_0$ becomes large in the interval $[E_{\mathrm{max}}, E_0-\epsilon]$, and we have omitted the second term, which is positive. Then, by using the Taylor expansion, we have
\begin{equation}\label{eq:beta_proof_6}
\begin{aligned}
    R(x)&=p_0\beta\epsilon+p_0\sum_{k=2}^\infty \frac{1}{k!}\left(\left(-\beta\epsilon\right)^k-(-2\beta\epsilon)^k \right)\\
    &>p_0\beta\epsilon.
\end{aligned}
\end{equation}
Again, using $\beta\epsilon\leq1$, it is simple to show that the second term in
Eq.~\eqref{eq:beta_proof_6} adds up to a positive factor because the positive terms are always larger than the adjacent negative terms in magnitude.
\end{proof}

Note that Lemma \ref{lemma:beta_value} basically says that in order to locate the ground state energy within desired accuracy, the resolution we want to achieve is $\epsilon^\prime=p_0\beta\epsilon$. This precision decides the complexity of the algorithm, which affects the performance of the dequantized algorithm that is discussed in Sec.~\ref{sec:deq}.

Now, let us discuss the quantum implementation for estimating the partition function given by Eq.~\eqref{eq:partition_function}. The overall idea is also to perform Fourier transformation to the partition function, resulting in a similar quantum execution as discussed in Sec.~\ref{sec:gsee_dev}. One caveat is that the exponential function is a not bounded function, which may cause trouble to the evaluation of the partition function. We manage to circumvent this problem by observing that $D_\beta(H-x)$ admits $x\in [E_a,E_0]$ in the process of the algorithm execution so that $e^{-\beta(H-x)}\preccurlyeq0$ is a trace non-increasing operator. In light of this observation, we can equivalently view the operator as $e^{-\beta|H-x|}$ with the function being altered to $e^{-\beta |x|}$. This insight has also been adopted for solving (dissipative) ordinary differentiable equations in Ref.~\cite{an2023linear}. Here, we remark that we do not directly apply this absolute-value exponential function because it is non-analytic, which could be challenging for dequantization using cluster expansion. 

Now, following from the equivalence between the convolution and partition function, we have the following result
\begin{equation}\label{eq:exp_fourier}
    D_\beta(H-x)=\int_{-\infty}^\infty\mathrm{~d}t \frac{\beta}{\pi\left(\beta^2+t^2\right)}  \bra{\psi_I}e^{- i (H-x)t}\ket{\psi_I},~(H-x)\succeq0.
\end{equation}
In Eq.~\eqref{eq:exp_fourier}, we have used the following Fourier transformation for the exponential function $g(x)=e^{-\beta|x|}$:
\begin{equation}
    \hat{g}(t)=\frac{1}{2 \pi} \int_{-\infty}^\infty e^{-\beta|x|} e^{-i t x} \mathrm{~d} x=\frac{\beta}{\pi\left(\beta^2+t^2\right)},
\end{equation}
where $\hat{g}(t)$ is the Cauchy-Lorentz distribution function. We then truncate the integral to approximate the partition function
\begin{equation}\label{eq:trunc_exp_fourier}
\begin{aligned}
    D_\beta(H-x)&\simeq \int_{-T}^T \mathrm{~d}t\frac{\beta}{\pi\left(\beta^2+t^2\right)}  \bra{\psi_I}e^{- i (H-x)t}\ket{\psi_I}\\
    &=\mathbb{E}_{t\sim\hat{g}_T^\prime(t)} [ \mathbf{Z}(x)],
\end{aligned}
\end{equation}
where $\hat{g}_T(x)$ is the Cauchy-Lorentz distribution truncated to the domain $[-T,T]$, $\hat{g}_T^\prime(x)=\hat{g}_T(x)/\|\hat{g}_T\|$, $\|\hat{g}_T\|$ is the normalization factor and $\mathbf{Z}(x):=\|\hat{g}_T\|\bra{\psi_I}e^{- i(H-x)t}\ket{\psi_I}$. Using the Hadamard test circuits, we then estimate the real and imaginary parts of $\mathbf{Z}(x,t)$. As such, we have the following two quantities
\begin{equation}
    \mathbb{E}\left[\mathbf{X}(t)\right]=\rm{Re}\left(\mathrm{tr}\left[\rho e^{-i Ht}\right]\right),\quad\mathbb{E}\left[\mathbf{Y}(t)\right]=\rm{Im}\left(\mathrm{tr}\left[\rho e^{-i Ht}\right]\right),
\end{equation}
where $\rm{Re}(\cdot)$ and $\rm{Im}(\cdot)$ denote taking the real and imaginary parts, separately. It is straightforward to verify the two parts comprise the unbiased estimator for Eq.~\eqref{eq:trunc_exp_fourier}: $\|\hat{g}_T\|e^{i xt}\left(\mathbb{E}\left[\mathbf{X}(t)\right]+\mathbb{E}\left[\mathbf{Y}(t)\right]\right)=\mathbb{E}\left[\mathbf{Z}(x)\right]$, where $\|\hat{g}_T\|e^{i xt}$ can be added through classical postprocessing.

\subsection{Complexity of the RQITE algorithm}
\label{sec:gsee_complexity}
In this section, we analyze the complexity of our RQITE algorithm in both maximal and total evolution times for the RQITE algorithm to solve the GSEE problem. We use the results from Ref.~\cite{wang2023quantum} for the sample complexity and refer the readers to [Lemma B.2, Ref.~\cite{wang2023quantum}] for detailed proof.
\begin{lemma}[Lemma B.2, Ref.~\cite{wang2023quantum}]\label{lemma:z_sample_complexity}
Let $\{t_i, \bm X(t_i), \bm Y(t_i)\}_{i=1}^S$ be the set of i.i.d.~samples such that $t_i\sim \hat{g}^\prime_T(t)$ and $\bm X(t_i)$ and $\bm Y(t_i)$ are measurement outcome from the Hadamard circuits.
Let also $\{x_j\}_{j=1}^M$ be the set of different points we need to evaluate for $D_\beta(H-x)$.
For each $j\in[M]$, define
\begin{equation}\label{eq:omega_estimator_mean}
    \overline{\bm{Z}}(x_j):=\frac{\|\hat{g}_T\|}S\sum_{i=1}^Se^{i t_i x_j}\left(\bm X(t_i)+i \bm Y(t_i)\right).
\end{equation}
Then for any $\epsilon^\prime>0$ and $\mu\in(0,1)$, when the total number of samples $S$ is given by
\begin{equation}\label{eq:z_sample}
    S:=\left\lceil \frac{\|\hat{g}_T\|^2\ln(4M/\mu)}{\epsilon^{\prime 2}} \right\rceil,
\end{equation}
we have
\begin{equation}\label{eq:sample_complexity}
    \mathbb{P}\left[\forall i\in[M],j\in[S]:|\overline{\textbf{Z}}(x_j)-\mathbb{E}[\mathbf{Z}(x_j)]|<\epsilon^\prime \right]\geq 1-\mu.
\end{equation}
\end{lemma}

We are ready to provide the overall complexity analysis of the RQITE algorithm for the GSEE tasks, which we state in the following theorem.
\begin{theorem}[RQITE algorithm for the GSEE tasks]\label{theorem:gsee_main_theorem}
Under the assumptions stated in Definition \ref{def:gsee} and algorithmic setting in Lemma~\ref{lemma:beta_value}, we can estimate the ground-state energy within sufficiently small additive error $\epsilon$, and probability at least $1-\mu$ using resources such that
\begin{itemize}
    \item The maximal Hamiltonian evolution time is $t_{\text{max}}=\mathcal{O}(\epsilon^{-1}\gamma^{-2})$;
    \item The total Hamiltonian evolution time is $t_{\text{total}}=\mathcal{\widetilde{O}}(\epsilon^{-3}\Delta^2\gamma^{-6})$.
    \item The classical postprocess time is $\mathcal{\widetilde{O}}(\epsilon^{-4}\Delta^2\gamma^{-6})$.
\end{itemize}
\end{theorem}
\begin{proof}
First, the maximal Hamiltonian evolution time is determined by the truncation $T$ in Eq.~\eqref{eq:trunc_exp_fourier}. 
We observe that 
\begin{eqnarray}
\begin{split}
    \left|\int_{-\infty}^{+\infty} \mathrm{~d}t\frac{\beta}{\pi\left(\beta^2+t^2\right)}e^{ixt}-\int_{-T}^T \mathrm{~d}t\frac{\beta}{\pi\left(\beta^2+t^2\right)}e^{ixt}\right|&=\left|\int_{-\infty}^{-T} \mathrm{~d}t\frac{\beta}{\pi\left(\beta^2+t^2\right)}e^{ixt} + \int_{T}^{+\infty} \mathrm{~d}t\frac{\beta}{\pi\left(\beta^2+t^2\right)}e^{ixt}\right|\\
    &= \int_{T}^{+\infty}\mathrm{~d}t \frac{\beta}{\pi (\beta^2 + t^2)} \left( e^{ixt} + e^{-ixt} \right) \\
    &\leq 2 \int_{T}^{+\infty}\mathrm{~d}t \frac{\beta}{\pi (\beta^2 + t^2)} |\cos(xt)|\\
    &\leq 2 \int_{T}^{+\infty}\mathrm{~d}t \frac{\beta}{\pi (\beta^2 + t^2)}\\
    &=1-\frac{2}{\pi} \arctan(T/\beta).\\
\end{split}
\end{eqnarray}
In the second line, we have used the symmetric property of the Cauchy distribution and let $x=-x$ for the first integral in the first line. The third line, we have used the trigonometry identity $(e^{ixt} + e^{-ixt})=2\cos(xt)$ and the fact that $\cos(xt)\leq |\cos(xt)|$. The fourth line is followed by $|\cos(xt)|\leq1$. Finally, we used the cumulant distribution function of the Cauchy distribution as $f(x)=\frac{1}{\pi} \arctan \left(\frac{x}{\beta}\right)+\frac{1}{2}$. 

Consequently, to guarantee an upper-bounded $\varepsilon$ truncation error, we have the truncation order to be
\begin{eqnarray}
\begin{split}
    T&\geq \beta \tan \left(\frac{\pi}{2}(1-\varepsilon)\right)\\
    &=\beta \frac{1}{\tan\left(\frac{\pi\varepsilon}{2}\right)}\approx \frac{2\beta}{\pi\varepsilon},
    \end{split}
\end{eqnarray}
where in the second line, we have used the identity $\tan(\frac{\pi}{2}-x)=\frac{1}{\tan(x)}$, the Taylor expansion that $\frac{1}{\tan(x)}=\frac{1}{x+x^3/3+\mathcal{O}(x^5)}$ and we truncate the approximation to it first order for sufficiently small $x$.

As such, we invoke [Claim A.5., Ref.~\cite{wang2023quantum}], which states that the truncation error in the convolution error is bounded by the truncation error induced by the truncation of the inverse Fourier transformation. This in turn gives us the error in truncating the partition function in Eq.~\eqref{eq:trunc_exp_fourier}:
$$\left|D_\beta(H-x)- \int_{-T}^T \mathrm{~d}t\frac{\beta}{\pi\left(\beta^2+t^2\right)}  \bra{\psi_I}e^{- i (H-x)t}\ket{\psi_I}\right|\leq \varepsilon.$$
Recall from Lemma \ref{lemma:beta_value} that the termination criteria is to determine the difference of two partition functions to be larger or smaller than a threshold of order $p_0\beta\epsilon$, which indicates that the accuracy for estimation of the partition function should be $\varepsilon= p_0\beta\epsilon/2$. This decides the truncation to be
\begin{eqnarray}
    T=\frac{4}{\pi\gamma^2\epsilon},
\end{eqnarray}
so that maximal evolution time satisfies $t_{\text{max}}=\mathcal{O}(\epsilon^{-1}\gamma^{-2})$. 

Next, we analyze the sample complexity using Lemma \ref{lemma:z_sample_complexity}: Given an $\mathbf{Z}(x_j)$ and accuracy $\epsilon^\prime$, the number of samples is $\mathcal{\widetilde{O}}\left(\epsilon^{\prime -2}\ln\left(\frac{4}{\mu\epsilon }\right)\right)$. From Lemma \ref{lemma:beta_value}, by letting $\epsilon^\prime=p_0\beta\epsilon/2$, we have bounded by the error from truncation of the integral and shot noise as $\varepsilon+\epsilon^\prime=p_0\beta\epsilon$.
Using the triangle inequality, we see that the total error is 
This gives the sample complexity $S=\mathcal{\widetilde{O}}(\epsilon^{-2}\Delta^2\gamma^{-4})$. The total evolution time is determined by the total evolution time and the sample complexity, which gives $\mathcal{\widetilde{O}}(\epsilon^{-3}\Delta^2\gamma^{-6})$.

The number of different $x$ we need to evaluate for $D_\beta(H-x)$ is given by $M=\mathcal{O}(\epsilon^{-1})$ because we run the algorithm for an interval $\epsilon$. This will contribute to the classical postprocessing time for evaluating the partition function at different locations. Besides, this part will contribute to the total number $M$ in Lemma \ref{lemma:z_sample_complexity}. 
Hence, the classical running time is larger than the sample complexity of quantum circuits by $\epsilon^{-1}$ so that we get $\mathcal{\widetilde{O}}(\epsilon^{-4}\Delta^2\gamma^{-6})$.
\end{proof}

Consider Assumption \ref{assume:gap_accuracy_relation}, the RQITE algorithm achieves an almost Heisenberg-limited scaling on $\epsilon$ for the total evolution time. Yet, the dependence on the overlap is worse because of the heavy tail of the Cauchy distribution. To compensate for this, one might think of other functions with more concentrated Fourier-transformed outcomes. For practical concerns, one possible solution is to multiply the Cauchy distribution with a Gaussian function as in Ref.~\cite{huo2023error}. Yet, as our final goal is to provide a quantized algorithm, the exponential function is important to us for the utility of cluster expansion~\cite{haah2024learning,wild2023classical,wu2024efficient}. The dequantized algorithm is affected by the large truncation order $T$ as we will see since it avoids using the convolution theorem.

\section{Classical Algorithm for approximating $D_{\beta}(H)$}

\subsection{Cluster and Interaction Graph}\label{sec:cluster}
\begin{definition}[Local Hamiltonian]
    A local Hamiltonian is composed of linear combinations of Hermitian operators $h_X$ which nontrivially acts on the qubit subset $X\in S$ with the corresponding coefficient $\lambda_X$. Here, the coefficients satisfy $\abs{\lambda_X}\leq 1$ and are chosen such that $\|h_X\|=1$. All subsystems set $X\in S$ are local operators. We define the associated Hamiltonian as $H=\sum_{X\in S}\lambda_Xh_X$.
\end{definition}

In this article, we assume Hermitian terms $h_X$ are distinct and non-identity multi-qubit Pauli operators. Such an assumption naturally satisfies $\|h_X\|=1$. computation.

\begin{definition}[Cluster induced by Hamiltonian]
    Given a local Hamiltonian $$H=\sum_{X\in S}\lambda_Xh_X,$$ a cluster $\bm{V}$ is defined as a nonempty multi-set of subsystems from $S$, where multi-sets allow an element appearing multiple times. The set of all clusters $\bm V$ with size $m$ is denoted by $\mathcal{C}_m$ and the set of all clusters is represented by $\mathcal{C}=\cup_{m\geq 1}\mathcal{C}_m$.
\end{definition}

For example, if the Hamiltonian $H=X_0X_1+Y_0Y_1$, then some possible candidates for $\bm V$ would be $\{X_0X_1\}$, $\{Y_0Y_1\}$, $\{X_0X_1,X_0X_1\},\cdots$. We call the number of times a subsystem $X$ appears in a cluster $\bm{V}$ the multiplicity $\mu_{\bm{V}}(X)$, otherwise we assign $\mu_{\bm{V}}(X)=0$. Traversing all subsets $X\in S$ may determine the size of $\bm{V}$, that is $\abs{\bm{V}}=\sum_{X\in S}\mu_{\bm V}(X)$. In the provided example, when $\bm V=\{X_0X_1,X_0X_1\}$, we have $\mu_{\bm V}(X_0X_1)=2$, $\mu_{\bm V}(Y_0Y_1)=0$ and $\abs{\bm V}=2$. 

\begin{definition}[Interaction Graph]
    We associate with every cluster $\bm{V}$ a simple graph $G_{\bm{V}}$ which is termed as the interaction graph. The vertices of $G_{\bm{V}}$ correspond to the subsystems in $\bm{V}$, with repeated subsystems also appearing as repeated vertices. Two distinct vertices $X$ and $Y$ are connected by an edge if and only if the respective subsystems overlap, that is ${\rm supp}(h_X)\cap{\rm supp}(h_Y)\neq\emptyset$.
\end{definition}

Suppose the cluster $\bm V=\{X_0X_1,X_0X_1\}$, then its corresponding interaction graph $G_{\bm V}$ has two vertices $v_1,v_2$, related to $X_0X_1$ and $X_0X_1$, respectively, and $v_1$ connects to $v_2$ since ${\rm supp}(X_0X_1)\cap {\rm supp}(X_0X_1)\neq \emptyset$. We say a cluster $\bm{V}$ is connected if and only if $G_{\bm{V}}$ is connected. We use the notation $\mathcal{G}_m$ to represent all connected clusters of size $m$ and $\mathcal{G}=\cup_{m\geq 1}\mathcal{G}_m$ for the set of all connected clusters.

In the context of the interaction graph, we utilize the notation $\mathfrak{d}$ to represent the maximum degree of the interaction graph $G_{\bm V}$. More specifically, if all terms contained in the cluster $\bm V$ are $k$-local, the maximum degree can be upper bounded by $$\mathfrak{d}\leq \max_{h_X\in S}4^k{\abs{\partial h_X}\choose k},$$ where $\abs{\partial h_X}$ represents the number of qubits on the boundary of $h_X$. When the Hamiltonian living in a $D$-dimensional lattice, it is shown that the $k$-qubit operator $h_X$ has at most $D2^{D-1}$ sides, where each side has $\mathcal{O}(k^{1/D})$ qubits. In the $D$-dimensional lattice, each qubit may directly connect to other $2D$ qubits, then we can evaluate 
\begin{align}
    \abs{\partial h_X}\leq 2^Dk^{1/D}D^2.
\end{align}
This may yield the maximum degree of the interaction graph
\begin{align}
    \mathfrak{d}\leq\mathcal{O}(4^k2^{Dk}k^{k/D}D^{2k}).
    \label{Eq: maximumdegree}
\end{align}

\subsection{Classical Algorithm for small $\beta$}
\label{sec:small_cluster}
We first consider the cluster expansion of $D_{\beta}(H)=\langle\psi_I|e^{-\beta H}|\psi_I\rangle$. Using the Taylor expansion formula, we have
\begin{align}
    e^{-\beta H}=\sum\limits_{m\geq 0}\frac{\beta^m}{m!}\left(\frac{\partial^m e^{-\beta H}}{\partial\beta^m}\right)_{\beta=0}.
    \label{Eq:origional}
\end{align}
Recall that $H=\sum_X\lambda_Xh_X$, then we define $Z_X=-\beta\lambda_X$ and $Z=(Z_{X_1},Z_{X_2},\cdots)$. As a result, we may write 
\begin{eqnarray}
\begin{split}
     \left(\frac{\partial^m e^{-\beta H}}{\partial\beta^m}\right)_{\beta=0}&=\sum\limits_{X_1,\cdots,X_m}\left(\frac{\partial Z_{X_1}}{\partial\beta}\right)\cdots \left(\frac{\partial Z_{X_m}}{\partial\beta}\right)\frac{\partial^m e^{-\beta H}}{\partial Z_{X_1}\cdots\partial Z_{X_m}}\big|_{Z=(0,0,\cdots,0)}\\
     &=\sum\limits_{X_1,\cdots,X_m}(-1)^m\lambda_{X_1}\cdots\lambda_{X_m}\frac{\partial^m}{\partial Z_{X_1}\cdots\partial Z_{X_m}}\sum\limits_{j\geq0}\frac{(-\beta)^j}{j!}H^j\big|_{Z=(0,0,\cdots,0)}\\
     &=\sum\limits_{X_1,\cdots,X_m}(-1)^m\lambda_{X_1}\cdots\lambda_{X_m}\frac{\partial^m}{\partial Z_{X_1}\cdots\partial Z_{X_m}}\sum\limits_{j\geq0}\frac{(-\beta)^j}{j!}\left(\sum\limits_{\lambda_1,\cdots,\lambda_j}\lambda_1\cdots\lambda_jh_{X_1}\cdots h_{X_j}\right)\big|_{Z=(0,0,\cdots,0)}\\
     &=\sum\limits_{X_1,\cdots,X_m}(-1)^m\lambda_{X_1}\cdots\lambda_{X_m}\frac{1}{m!}\sum\limits_{\sigma\in\mathcal{P}_m}h_{X_{\sigma(1)}}\cdots h_{X_{\sigma(m)}}.
     \label{Eq:partialderevate}
\end{split}
\end{eqnarray}
Taking Eq.~\ref{Eq:partialderevate} into Eq.~\ref{Eq:origional}, we finally obtain the cluster expansion formula
\begin{eqnarray}
    \begin{split}
         D_{\beta}(H)&=\sum\limits_{m\geq 0}\frac{\beta^m}{m!}\sum\limits_{X_1,\cdots,X_m}(-1)^m\lambda_{X_1}\cdots\lambda_{X_m}\frac{1}{m!}\sum\limits_{\sigma\in\mathcal{P}_m}{\rm Tr}\left[|\psi_c\rangle\langle\psi_c|h_{X_{\sigma(1)}}\cdots h_{X_{\sigma(m)}}\right]\\
         &=\sum\limits_{m\geq 0}\sum\limits_{\bm W\in\mathcal{C}_m}\frac{(-\beta)^{\abs{\bm W}}{\bm\lambda}^{\bm W}}{\bm W!}\frac{1}{m!}\sum\limits_{\sigma\in\mathcal{P}_m}{\rm Tr}\left[|\psi_c\rangle\langle\psi_c|h_{X_{\sigma(1)}}\cdots h_{X_{\sigma(m)}}\right]\\
         &=1+\sum\limits_{m\geq 1}\sum\limits_{\bm W\in\mathcal{C}_m}\frac{(-\beta)^{\abs{\bm W}}{\bm\lambda}^{\bm W}}{\bm W!}\frac{1}{m!}\sum\limits_{\sigma\in\mathcal{P}_m}{\rm Tr}\left[|\psi_c\rangle\langle\psi_c|h_{X_{\sigma(1)}}\cdots h_{X_{\sigma(m)}}\right],
    \end{split}
\end{eqnarray}
where $\bm\lambda^{\bm W}=\prod_{X\in S}\lambda_W^{\mu_{\bm W}(X)}$ and $\bm W!=\prod_{X\in S}\mu_{\bm W}(X)!$. 

From the above expression, it is shown that each cluster $\bm W=(h_{X_1},\cdots,h_{X_m})$, and ${\rm Tr}\left[|\psi_c\rangle\langle\psi_c|h_{X_{\sigma(1)}}\cdots h_{X_{\sigma(m)}}\right]$ can be factorized when $\bm W$ is disconnected and $|\psi_c\rangle$ be the product state. In general, the classical initial state $|\psi_I\rangle$ may not represent a product state, but a superposition quantum state with $R\leq {\rm poly}(n)$ configurations. Suppose $|\psi_I\rangle=\sum_xa_x|x\rangle$ where $|x\rangle$ represents a tensor product state and coefficients satisfy $\sum_x\abs{a_x}^2=1$. As a result, the partition function can be decomposed by $$D_{\beta}(H)=\sum_{x,y}a_xa_y^{*}\langle y|e^{-\beta H}|x\rangle.$$ In the following sections, we focus on estimating $d_{x,y,\beta}(H)=\langle y|e^{-\beta H}|x\rangle$.

When the cluster $\bm W$ is disconnected, the related term in the cluster expansion of $d_{x,y,\beta}(H)$ is given by
$$\frac{1}{m!}\sum_{\sigma\in\mathcal{P}_m}{\rm Tr}\left[|x\rangle\langle y|h_{X_{\sigma(1)}}\cdots h_{X_{\sigma(m)}}\right]=\prod\limits_{\bm V\in P_{c,\max}(\bm W)}\langle y|h^{\bm V}|x\rangle,$$
with $P_{c,\max}(\bm W)$ representing the partition of the cluster $\bm W$ into its maximal connected components. For the convenience of the following description, we denote $\mathcal{D}_{\bm W}(d_{x,y,\beta}(H))=(-\beta)^{\bm W}\prod\limits_{\bm V\in P_{c,\max}(\bm W)}\langle y|h^{\bm V}|x\rangle$, which naturally give rises to
\begin{align}
    d_{x,y,\beta}(H)=1+\sum\limits_{m\geq 1}\sum\limits_{\bm W\in\mathcal{C}_m}\prod\limits_{\bm V\in P_{c,\max}(\bm W)}\frac{\bm\lambda^{\bm V}}{\bm V!}\mathcal{D}_{\bm V}(d_{x,y,\beta}(H)).
\end{align}

Directly compute $d_{x,y,\beta}(H)$ is generally hard, however, we can compute the function $\log(d_{x,y,\beta}(H))$ when the inverse temperature $\beta$ is smaller than a threshold. We consider the formal Taylor series 
\begin{align}
    \log(1+z)=\sum\limits_{k=1}^{\infty}\frac{(-1)^{k-1}}{k}z^k.
\end{align}
Taking $d_{x,y,\beta}(H)$ into above series may yield
\begin{eqnarray}
    \begin{split}
        \log\left(d_{x,y,\beta}(H)\right)=\sum\limits_{k=1}^{\infty}\frac{(-1)^k}{k}\sum\limits_{m_1,\cdots,m_k}\sum\limits_{\bm W_1\in\mathcal{C}_{m_1},\cdots,\bm W_k\in\mathcal{C}_{m_k}}&\left(\prod\limits_{\bm V_1\in P_{c,\max}(\bm W_1)}\frac{\lambda^{\bm V_1}}{\bm V_1!}\mathcal{D}_{\bm V_1}(d_{x,y,\beta}(H))\right)\cdots\\
        &\cdots\left(\prod\limits_{\bm V_k\in P_{c,\max}(\bm W_k)}\frac{\lambda^{\bm V_k}}{\bm V_1!}\mathcal{D}_{\bm V_k}(d_{x,y,\beta}(H))\right).
    \end{split}
    \label{Eq:logcluster}
\end{eqnarray}
Using the similar method mentioned in Ref.~\cite{wild2023classical}, we can re-group the sums over all clusters, that is
\begin{align}
    \log\left(d_{x,y,\beta}(H)\right)=\sum\limits_{m\geq 1}\sum\limits_{\bm W\in\mathcal{G}_m}\sum\limits_{P\in\mathcal{P}_c(\bm W)}C(P)\prod\limits_{\bm V\in P}\frac{\lambda^{\bm V}}{\bm V!}\mathcal{D}_{\bm V}(d_{x,y,\beta}(H)).
\end{align}
Here, $\mathcal{G}_m$ represents the set of all connected clusters $\bm W$ with size $m$, $\mathcal{P}_c(\bm W)$ represents the set of all partitions $P$ to the cluster $\bm W$, and the coefficient $C(P)$ can be determined by considering the different ways in which the partition $P$ can be generated by the clusters $\bm W_1,\cdots,\bm W_k$ given by Eq.~\ref{Eq:logcluster}.

\begin{lemma}[Proposition~9 in Ref.~\cite{wild2023classical}]
    Let the cluster $\bm W\in\mathcal{G}_m$, then we have
    \begin{align}
        \abs{\sum\limits_{P\in\mathcal{P}_c(\bm W)}C(P)\prod\limits_{\bm V\in P}\frac{\lambda^{\bm V}}{\bm V!}\mathcal{D}_{\bm V}(d_{x,y,\beta}(H))}\leq \left[2e(\mathfrak{d}+1)\abs{\beta}\right]^m.
    \end{align}
    \label{lemma:error1}
\end{lemma}
Above lemma enables us to approximate $\log\left(d_{x,y,\beta}(H)\right)$ by truncating the cluster expansion up to $M\leq\mathcal{O}(\log(\abs{S}R/\epsilon))$ order when $\abs{\beta}<1/2e^2\mathfrak{d}(\mathfrak{d}+1)$, that is
\begin{eqnarray}
    \begin{split}
        &\abs{\log\left(d_{x,y,\beta}(H)\right)-\sum\limits_{m=1}^M\sum\limits_{\bm W\in\mathcal{G}_m}\sum\limits_{P\in\mathcal{P}_c(\bm W)}C(P)\prod\limits_{\bm V\in P}\frac{\lambda^{\bm V}}{\bm V!}\mathcal{D}_{\bm V}(d_{x,y,\beta}(H))}\\
        =&\abs{\sum\limits_{m\geq M+1}\sum\limits_{\bm W\in\mathcal{G}_m}\sum\limits_{P\in\mathcal{P}_c(\bm W)}C(P)\prod\limits_{\bm V\in P}\frac{\lambda^{\bm V}}{\bm V!}\mathcal{D}_{\bm V}(d_{x,y,\beta}(H))}\\
        \leq&\sum\limits_{m\geq M+1}\sum\limits_{\bm W\in\mathcal{G}_m}\left[2e(\mathfrak{d}+1)\abs{\beta}\right]^m\\
        \leq&\sum\limits_{m\geq M+1}\abs{S}\left[2e^2\mathfrak{d}(\mathfrak{d}+1)\abs{\beta}\right]^m\\
        =&\frac{\abs{S}\left[2e^2\mathfrak{d}(\mathfrak{d}+1)\abs{\beta}\right]^{M+1}}{1-\left[2e^2\mathfrak{d}(\mathfrak{d}+1)\abs{\beta}\right]}.
    \end{split}
\end{eqnarray}
The third line comes from Lemma~\ref{lemma:error1}, and the fourth line is valid since $\abs{\mathcal{G}_m}\leq \abs{S}(e\mathfrak{d})^m$, with $\abs{S}$ represents the number of local terms in the Hamiltonian $H$. Let 
\begin{align}
    \epsilon=\frac{\abs{S}\left[2e^2\mathfrak{d}(\mathfrak{d}+1)\abs{\beta}\right]^{M+1}}{1-\left[2e^2\mathfrak{d}(\mathfrak{d}+1)\abs{\beta}\right]},
\end{align}
and this directly yields 
\begin{align}
    M\leq\frac{\log\left(\frac{\abs{S}}{\epsilon[1-\left[2e^2\mathfrak{d}(\mathfrak{d}+1)\abs{\beta}\right]]}\right)}{\log(1/\left[2e^2\mathfrak{d}(\mathfrak{d}+1)\abs{\beta}\right])}.
\end{align}
As a result, we can evaluate the running time complexity for computing the truncated $M$-order Taylor series. It is shown that the connected cluster set satisfies $\abs{\mathcal{G}_m}\leq \abs{S}(e\mathfrak{d})^m$. Recall $\mathcal{P}_c(\bm W)$ represents all connected partitions given by the cluster $\bm W=(\bm W_1,\cdots,\bm W_m)$, and enumerating all partitions of $\bm W$ into connected subclusters takes time $\exp(\mathcal{O}(M))$ (Details refer to Proposition~11 in Ref.~\cite{wild2023classical}). Finally, the coefficient $C(P)$ can be computed in time $\exp(\mathcal{O}(\abs{P}))$ by using the algorithm given by Ref.~\cite{bjorklund2008fast}. Taking all together, we can summarize that there exists a $$\abs{S}\exp(\mathcal{O}(M))=\abs{S}{\rm poly}\left[\left(\frac{\abs{S}}{\epsilon[1-\left[2e^2\mathfrak{d}(\mathfrak{d}+1)\abs{\beta}\right]}\right)^{\log(1/\left[2e^2\mathfrak{d}(\mathfrak{d}+1)\abs{\beta}\right])}\right]$$ running time classical algorithm that can output an approximation to $\log\left(d_{x,y,\beta}(H)\right)$ within $\epsilon$ additive error. Equivalently, we can output an approximation $\hat{d}_{x,y,\beta}(H)$ such that 
\begin{align}
    e^{-\epsilon}\abs{\hat{d}_{x,y,\beta}(H)}\leq\abs{d_{x,y,\beta}(H)}\leq e^{\epsilon}\abs{\hat{d}_{x,y,\beta}(H)}.
\end{align}
This implies 
\begin{align}
    \frac{\abs{\hat{d}_{x,y,\beta}(H)-d_{x,y,\beta}(H)}}{\abs{d_{x,y,\beta}(H)}}=\abs{\frac{\hat{d}_{x,y,\beta}(H)}{d_{x,y,\beta}(H)}-1}\approx\abs{\log\left(\frac{\hat{d}_{x,y,\beta}(H)}{d_{x,y,\beta}(H)}\right)}=\abs{\log(\hat{d}_{x,y,\beta}(H))-\log(d_{x,y,\beta}(H))}\leq\epsilon,
\end{align}
where the second approximation is valid since $e^{-\epsilon}\leq \abs{\frac{\hat{d}_{x,y,\beta}(H)}{d_{x,y,\beta}(H)}}\leq e^{\epsilon}$. As a result, we can utilize estimators $\hat{d}_{x,y,\beta}(H)$ to construct an estimator to $D_{\beta}(H)$, such that
\begin{eqnarray}
    \begin{split}
        \abs{D_{\beta}(H)-\hat{D}_{\beta}(H)}&=\abs{\sum\limits_{x,y}a_xa^*_y\left(d_{x,y,\beta}(H)-\hat{d}_{x,y,\beta}(H)\right)}\\
        &\leq\sum\limits_{x,y}\abs{a_xa_y^*}\abs{d_{x,y,\beta}(H)-\hat{d}_{x,y,\beta}(H)}\\
        &\leq \epsilon\sum\limits_{x,y}\abs{a_xa_y^*}\abs{d_{x,y,\beta}(H)}\\
        &\leq \epsilon\max_{x,y}\abs{d_{x,y,\beta}(H)}\\
        &\leq \epsilon e^{-\beta E_0/2},
    \end{split}
\end{eqnarray}
where $E_0$ represents the ground state energy of $H$. 
Finally, we notice that Eq.~\eqref{eq:partition_function} gives the targeted partition function, where a shift $x$ is introduced. The value of $x$ is determined by the RQITE algorithm as given by the interval $[E_a, E_0]$. This gives us 
\begin{eqnarray}
    \abs{D_{\beta}(H-x)-\hat{D}_{\beta}(H-x)}\leq \epsilon e^{-\beta (E_0-x)/2}\leq \epsilon,~x\in[E_a,E_0].
\end{eqnarray}
The above results can be summarized as follows:
\begin{theorem}\label{thm:partition_cluster}
    Given a local Hamiltonian $H$ with $\abs{S}$ local terms, a classical state $|\psi_c\rangle$ with $R$ configurations, arbitrary $x\in[E_a, E_0]$ and inverse temperature $0\leq \abs{\beta}<\beta^*=1/2e^2\mathfrak{d}(\mathfrak{d}+1)$, then there exists a classical algorithm that can output an estimator to $D_{\beta}(H-x)$, such that $\abs{\hat{D}_{\beta}(H-x)-D_{\beta}(H-x)}\leq \epsilon$ within
    \begin{align}\label{eq:partition_trunc_small}
        R^2\abs{S}{\rm poly}\left[\left(\frac{\abs{S}}{\epsilon[1-\abs{\beta}/\beta^*]}\right)^{\log(\beta^*/\abs{\beta})}\right]
    \end{align}
    classical running time, where $E_0$ represents the ground state energy of $H$.
\end{theorem}

\subsection{Extend the inverse temperature $\beta$ to a general constant}
\label{sec:analytic_continuation}
Following the analytic continuation method given by Ref.~\cite{wild2023classical}, we consider the map $\beta\mapsto\beta\phi(z)$, where the complex variable function $$\phi(z)=\frac{\log(1-z/\nu)}{\log(1-1/\nu^{\prime})}$$ satisfies (i) $\phi(0)=0$, (ii) $\phi(1)=1$ and (iii) $\phi(z)$ is analytic on a disk $\abs{z}\leq \nu$, where ${\rm Im}[\phi(z)]\leq w$ for some constant value $w$ and $1<\nu<\nu^{\prime}$. More specifically, let the complex function $f(z)=\log(D_{\beta\phi(z)}(H))$, and we consider utilizing the complex Taylor approximation method (Lemma~5 in Ref~\cite{wild2023classical}) to upper bound 
\begin{align}
    \epsilon_M(z)=\abs{f(z)-\sum\limits_{m=0}^M\frac{f^{(m)}(0)}{m!}z^m}\leq\frac{\alpha^{M+1}}{1-\alpha}\max_{z\in D_{\alpha \nu}}\abs{f(z)},
    \label{ineq:complextaylor}
\end{align}
where the parameter $\alpha\in(0,1)$. In the following paragraphs, we demonstrate how to upper bound $\abs{f(z)}$ on the closed disk $D_{\alpha \nu}$ for $\alpha\in(0,1)$. 

To do this, we assume that the initial state $|\psi_c\rangle$ can be prepared by a constant depth quantum circuit $U$, that is $|\psi_c\rangle=U|0^n\rangle$. Let the modified Hamiltonian $H^{\prime}=U^{\dagger}HU$, then we have
\begin{eqnarray}
    \begin{split}
        \max_{z\in D_{\alpha \nu}}\abs{\log(D_{\beta\phi(z)}(H))}&=\max_{z\in D_{\alpha \nu}}\abs{\log\left(\langle\psi_c|e^{-\beta\phi(z)H}|\psi_c\rangle
        \right)}\\
        &=\max_{z\in D_{\alpha \nu}}\abs{\log\left(\langle0^n|U^{\dagger}e^{-\beta\phi(z)H}U|0^n\rangle
        \right)}\\
        &=\max_{z\in D_{\alpha \nu}}\abs{\log\left(\langle0^n|e^{-\beta\phi(z)H^{\prime}}|0^n\rangle
        \right)}\\
        &=\max_{z\in D_{\alpha \nu}}\abs{\log\left(\langle0^n|e^{-{\rm Re}[\beta\phi(z)]H^{\prime}}e^{-{\rm Im}[\beta\phi(z)]H^{\prime}}|0^n\rangle\right)}.
    \end{split}
\end{eqnarray}
We denote the normalized quantum states $$|\phi\rangle=c^{-1/2}e^{-{\rm Re}[\beta\phi(z)]H^{\prime}}|0^n\rangle=\frac{e^{-{\rm Re}[\beta\phi(z)]H^{\prime}}|0^n\rangle}{\sqrt{\langle0^n|e^{-2{\rm Re}[\beta\phi(z)]H^{\prime}}}|0^n\rangle}=\sum_{x}\phi_x|x\rangle$$ and $$|\psi\rangle=e^{-{\rm Im}[\beta\phi(z)]H^{\prime}}|0^n\rangle|0^n\rangle=\sum_x\psi_x|x\rangle,$$ where $|x\rangle$ represents a set of computational basis. Without loss of generality, we assume $\abs{\phi_x}\geq 2^{-{\rm poly}(n)}$, $\abs{\psi_x}\geq 2^{-{\rm poly}(n)}$ or $\phi_x\neq 0$ and $\psi_x\neq 0$. Otherwise, the component may be negligible regarding the semi-classical state.

Using above notations, we upper bound $\max_{z\in D_{\alpha \nu}}\abs{\log(D_{\beta\phi(z)}(H))}$ from two perspectives. On one hand, if $\abs{\psi_x}<\abs{\phi_x}$, we have 
\begin{eqnarray}\label{eq:op_norm0}
    \begin{split}
         &\max_{z\in D_{\alpha \nu}}\abs{\log\left(\langle0^n|e^{-{\rm Re}[\beta\phi(z)]H^{\prime}}e^{-(\beta\phi(z)-{\rm Re}[\beta\phi(z)])H^{\prime}}|0^n\rangle\right)}\\
         =&\max_{z\in D_{\alpha \nu}}\abs{\log\left(c^{1/2}\langle\phi|e^{-i{\rm Im}[\beta\phi(z)]H^{\prime}}|0^n\rangle\right)}\\
         =&\max_{z\in D_{\alpha \nu}}\abs{\log\left(c^{1/2}\sum_{x}\phi_x^*\langle x|e^{-i{\rm Im}[\beta\phi(z)]H^{\prime}}|0^n\rangle\right)}\\
         \leq&\max_{z\in D_{\alpha \nu}}\abs{\log\left(c^{1/2}\right)}+\max_{z\in D_{\alpha \nu}}\abs{\log\left(\sum_{x}\abs{\phi_x}^2\frac{\psi_x}{\phi_x}\right)}\\
         \leq&\max_{z\in D_{\alpha \nu}}\frac{1}{2}\abs{\log(e^{-2{\rm Re}(\beta\phi(z))E_0})}+\max_{z\in D_{\alpha \nu}}\sum_{x}\abs{\phi_x}^2\abs{\log\left(\frac{\psi_x}{\phi_x}\right)}\\
         \leq& \max_{z\in D_{\alpha \nu}}\frac{1}{2}\abs{\log(e^{-2{\rm Re}(\beta\phi(z))E_0})}+ \max_{z\in D_{\alpha \nu},x}\abs{\log(\psi_x)}+\max_{z\in D_{\alpha \nu},x}\abs{\log(\phi_x)}.
    \end{split}
\end{eqnarray}
The fifth line is valid since $\abs{\log\left(\sum_{x}\abs{\phi_x}^2\frac{\psi_x}{\phi_x}\right)}=\abs{\log\left(\abs{\sum_{x}\abs{\phi_x}^2\frac{\psi_x}{\phi_x}}\right)+i\theta}$, where $\theta=\mathcal{O}(1)$ represents the phase of $z=\sum_{x}\abs{\phi_x}^2\frac{\psi_x}{\phi_x}$. We then take $m=2$ and $\vec{a}_k=\pm \vec{1}$ into Theorem~4 of Ref.~\cite{ansari2005more}, and we may lower bound $\abs{z}$ by $2^{-n}\sum_x\abs{\phi_x}^2\abs{\frac{\psi_x}{\phi_x}}$. Finally, apply the Jensen inequality to function $f(z)=\abs{\log2^{-n}\sum_x\abs{\phi_x}^2\abs{\frac{\psi_x}{\phi_x}}}$ yielding the fifth line, where we omit the phase $\theta$ and a term $n$.

On other hand, if $\abs{\psi_x}>\abs{\phi_x}$, the upper bound
\begin{eqnarray}\label{eq:op_norm0}
    \begin{split}
         &\max_{z\in D_{\alpha \nu}}\abs{\log\left(\langle0^n|e^{-{\rm Re}[\beta\phi(z)]H^{\prime}}e^{-(\beta\phi(z)-{\rm Re}[\beta\phi(z)])H^{\prime}}|0^n\rangle\right)}\\
         =&\max_{z\in D_{\alpha \nu}}\abs{\log\left(\langle0^n|e^{-{\rm Re}[\beta\phi(z)]H^{\prime}}|\psi\rangle\right)}\\
         =&\max_{z\in D_{\alpha \nu}}\abs{\log\left(\sum_{x}\psi_x\langle0^n|e^{-{\rm Re}[\beta\phi(z)]H^{\prime}}|x\rangle\right)}\\
         \leq&\max_{z\in D_{\alpha \nu}}\abs{\log\left(c^{1/2}\right)}+\max_{z\in D_{\alpha \nu}}\abs{\log\left(\sum_{x}\abs{\psi_x}^2\frac{\langle0^n|e^{-{\rm Re}[\beta\phi(z)]H^{\prime}}|x\rangle}{c^{1/2}\psi^{*}_x}\right)}\\
         \leq&\max_{z\in D_{\alpha \nu}}\frac{1}{2}\abs{\log(e^{-2{\rm Re}(\beta\phi(z))E_0})}+\max_{z\in D_{\alpha \nu}}\sum_{x}\abs{\psi_x}^2\abs{\log\left(\frac{\phi_x}{\psi^*_x}\right)}\\
         \leq& \max_{z\in D_{\alpha \nu}}\frac{1}{2}\abs{\log(e^{-2{\rm Re}(\beta\phi(z))E_0})}+ \max_{z\in D_{\alpha \nu},x}\abs{\log(\psi_x)}+\max_{z\in D_{\alpha \nu},x}\abs{\log(\phi_x)}.
    \end{split}
\end{eqnarray}
is still valid.

As a result, in both scenarios ($\abs{\psi_x}>\abs{\phi_x}$ or $\abs{\psi_x}\leq\abs{\phi_x}$), we have
\begin{eqnarray}\label{eq:op_norm_final}
    \begin{split}
       &\max_{z\in D_{\alpha \nu}}\abs{\log(D_{\beta\phi(z)}(H))}\\
       \leq&\max_{z\in D_{\alpha \nu}}\frac{1}{2}\abs{\log(e^{-2{\rm Re}(\beta\phi(z))E_0})}+ \max_{z\in D_{\alpha \nu},x}\abs{\log(\psi_x)}+\max_{z\in D_{\alpha \nu},x}\abs{\log(\phi_x)}\\
       \leq & \max_{z\in D_{\alpha \nu}}\frac{1}{2}\abs{\log(e^{-2{\rm Re}(\beta\phi(z))E_0})}+\sum\limits_{m\geq 1}(e\mathfrak{d})^m\sum\limits_{P\in\mathcal{P}_c(\bm W)}C(P)\prod\limits_{\bm V\in P}\frac{\lambda^{\bm V}}{\bm V!}\max_{z}\abs{-{\rm Im}(\beta\phi(z))}^{|\bm V|}+\max_{z\in D_{\alpha \nu},x}\abs{\log(\phi_x)}\\
       \leq & \abs{{\rm Re}(\beta\phi(z))E_0}+\abs{S}\frac{2e^2\mathfrak{d}(\mathfrak{d}+1)w\beta}{1-2e^2\mathfrak{d}(\mathfrak{d}+1)w\beta}+{\rm poly}(n).
    \end{split}
    \label{ineqmax}
\end{eqnarray}

Taking \ref{ineqmax} into \ref{ineq:complextaylor}, we immediately have
\begin{align}
    \epsilon(1)\leq\frac{\alpha^{M+1}}{1-\alpha}\left(\abs{\beta E_0}+\abs{S}\frac{2e^2\mathfrak{d}(\mathfrak{d}+1)w\beta}{1-2e^2\mathfrak{d}(\mathfrak{d}+1)w\beta}+{\rm poly}(n)\right).
\end{align}
Let $\alpha=1/\nu$, $w=\frac{-\pi}{2\log(1-1/\nu^{\prime})}$ with constant values $\nu, \nu^{\prime}$ satisfying $1<\nu<\nu^{\prime}$. Besides, for the choice of $\nu^\prime$, we let $w=0.5(\beta^*/\beta)$, where $\beta^*=1/2e^2\mathfrak{d}(\mathfrak{d}+1)$. Then, we finally have
\begin{align}
    M\leq e^{2\pi\beta/\beta^*}\log\left[\frac{ e^{2\pi\beta/\beta^*}}{\epsilon}\left(\abs{\beta E_0}+\rm{poly}(\abs{S})\right)\right].
\end{align}


From the above analysis, the analytic continuation method requires conditions: (i) ${\rm Im}(\phi(z))\leq w$, (ii) $w\beta<1/2e^2\mathfrak{d}(\mathfrak{d}+1)$, and (iii) $D_{\beta\phi(z)}(H)$ is zero free when $\abs{z}\leq\nu$. In Theorem~\ref{thm:zero-free}, we rigorously proved that ${\rm Re}(\beta)>0$ enable the zero free property of $D_{\beta\phi(z)}(H)$. Combine all constraints together, the analytic domain of $f(z)=\log(D_{\beta\phi(z)}(H))$ can be characterized by
\begin{eqnarray}
    \left\{
        \begin{aligned}
        -1/2e^2\mathfrak{d}(\mathfrak{d}+1)<&{\rm Im}[\beta\phi(z)]<1/2e^2\mathfrak{d}(\mathfrak{d}+1)\\
        0<&{\rm Re}[\beta\phi(z)].\\
        \end{aligned}
    \right.
\end{eqnarray}

\subsection{Classical Algorithm and Complexity Analysis}
\label{sec:arbitrary_complexity}
Here, we provide technical details in evaluating the function
\begin{align}\label{eq:target}
    f(1)=\log\left(D_{\beta\phi(1)}(H)\right)=\log\left(\langle0^n|e^{-\beta\phi(1)H^{\prime}}|0^n\rangle\right).
\end{align}
Let $\log\left(\langle0^n|e^{-\beta\phi(1)H^{\prime}}|0^n\rangle\right)=\sum_{l=0}^{\infty}A_l\beta^l$ and $\phi(z)=\sum_{l=0}^{\infty}\phi_lz^l$. According to the Taylor series in the complex plane, we may approximate $f(1)$ by
\begin{eqnarray}
    \begin{split}
         f(1)=&\sum\limits_{m=0}^M\frac{f^{(m)}(0)}{m!}\\
         =&\sum\limits_{m=0}^M\left(\sum\limits_{l=1}^mA_l\beta^l\sum\limits_{m_1,\cdots,m_l\geq 1, m_1+\cdots+m_l=k}\phi_{m_1}\cdots\phi_{m_l}\right),
    \end{split}
\end{eqnarray}
    where
    \begin{align}
        A_l=\sum\limits_{\bm W\in\mathcal{G}_l}(-1)^{\bm W}\sum\limits_{P\in\mathcal{P}_c(\bm W)}C(P)\prod\limits_{\bm V\in P}\frac{\lambda^{\bm V}}{\bm V!}\langle 0^n|h^{\bm V}|0^n\rangle
    \end{align}
is given by the cluster expansion method. As a result, the computational complexity is upper bounded by
\begin{align}\label{eq:analy_complex}
    \exp(M)\leq \left[\frac{e^{2\pi \beta/\beta^*}}{\epsilon}\left(\abs{\beta E_0}+\rm{poly}(\abs{S})\right)\right]^{e^{2\pi \beta/\beta^*}}
\end{align}

We conclude this section with the above results summarized by the following theorem. Besides, we focus on the case that $\rm{Re}(\beta)\geq 0$ to facilitate later discussion for the dequantization of the RQITE algorithm.
\begin{theorem}\label{thm:partition_analytic_continuation}
    Given a local Hamiltonian $H$ with $\abs{S}$ local terms, a classical state $|\psi_c\rangle$ with $R$ configurations, arbitrary $x\in[E_a, E_0]$, then there exists a classical algorithm that can output an estimator to $f(H-x)=\log(D_{\beta}(H-x))$, such that $\abs{\hat{f}(H-x)-f(H-x)}\leq \epsilon$ for $\beta\in\mathbb{C}$, $\rm{Re}(\beta)\geq 0$ within classical running time
    \begin{align}\label{eq:analytic_cost}
        \left[\frac{e^{2\pi \beta/\beta^*}}{\epsilon}\rm{poly}(\abs{S})\right]^{e^{2\pi \beta/\beta^*}},
    \end{align}
    where $\beta^*=1/2e^2\mathfrak{d}(\mathfrak{d}+1)$.
\end{theorem}
\begin{proof}
The change of subject from $\log (D_{\beta}(H))$ to $\log (D_{\beta}(H-x))$ with $x\in[E_a, E_0]$ will affect the operator norm of the subject, which plays a key role in the complex Taylor approximation method. To see the effect, the dependence on the ground-state energy of the operator norm originates from the $c^{1/2}=\sqrt{\langle0^n|e^{-2{\rm Re}[\beta\phi(z)]H^{\prime}}}|0^n\rangle$ term in Eq.~\eqref{eq:op_norm0}. As the variable of the partition function changes from $H$ to $(H-x)$, $c^{1/2}$ accordingly becomes 
\begin{eqnarray}
\begin{split}
    &\sqrt{\langle0^n|e^{-2{\rm Re}[\beta\phi(z=1)](H^{\prime}-x)}|0^n\rangle}\\
    =&\sqrt{\langle \psi_c|e^{-2{\rm Re}[\beta](H-x)}|\psi_c\rangle}\\
    \leq&\sqrt{\langle \psi_0|e^{-2{\rm Re}[\beta](E_0-x)}|\psi_0\rangle}\leq 1.
\end{split}
\end{eqnarray}
In the first line, we have taken $z=1$ for the estimation of $f(H-x)$ as $\phi(1)=1$. The second line utilizes the fact that $H^\prime=U^\dagger H U$ and $\ket{\psi_c}=U\ket{0^n}$. In the last line, we substitute $\ket{\psi_c}$ with $\ket{\psi_0}$, i.e.~the ground state of $H$, because $\rm Re[\beta](H-x)$ is positive semi-definite for all choices of $x$, which makes the ITE operator $e^{-2{\rm Re}[\beta](H-x)}$ trace non-increasing. As such, $c^{1/2}$ is bounded by $1$ and the dependence on the ground-state energy is dismissed.
\end{proof}

\section{Dequantization of the RQITE algorithm}\label{sec:deq}
In this section, we provide details on the dequantization of the RQITE algorithms. 

To extend the results to arbitrary constant accuracy, one could resort to analytic continuation of the partition function as in Ref.~\cite{wild2023classical}. Yet, the analytic continuation relies on the logarithm of the partition function to be analytic, i.e.~identifying regions of the partition function to be zero-free. Such a problem is known to be challenging from the studying of the partition function of the Gibbs states, which is $\operatorname{Tr}(e^{-\beta H})$. The connection between our partition function and that of the Gibbs states is that our partition function can be viewed as the expectation value of a positive semi-definite observable w.r.t.~the Gibbs state: $\operatorname{Tr}(\rho_I e^{-\beta H}),~\rho_I \succcurlyeq 0$, which is studied in Ref.~\cite{mann2024algorithmic}. What makes the situation even worse is that the hardness of approximation results suggests that there exists a constant temperature $\beta_c$ such that $\beta>\beta_c$ becomes classically intractable~\cite{sly2012computational,galanis2016inapproximability,goldberg2017complexity,mann2024algorithmic}. The implications here are that one cannot in general approximate the Gibbs state at arbitrary constant temperature $\beta$.

Despite such a grim situation, we identify that in certain situations, the zero-free regions of our partition function could be pinpointed so that an arbitrary constant accuracy could be achieved for dequantizing the RQITE algorithm. One situation lets us circumvent the problem mentioned above by observing that an analytic region could be located when the ground-state component $p_0$ in the initial state is greater than a certain threshold.

\subsection{Classical algorithm for solving the GSEE problem with limited accuracy}
\label{sec:deq_limit}
We are now in place to describe the dequantized RQITE algorithm. The fact that cluster expansion only allows a limited imaginary evolution time will constrain the accuracy we can reach, which is summarized in the following.

\begin{theorem}[Dequantization of the RQITE algorithm with limited accuracy]\label{thm:deq_rqiteapp}
Let $\ket{\psi_c}=\sum_{i=1}^R \alpha_i\ket{i}$, be the semi-classical state, where $R=\mathcal{O}(\text{poly}(n))$. For a targeted local Hamiltonian $H$, assume that $\ket{\psi_c}$ has a lower-bounded overlap $\gamma$ with the ground state of $H$. Then, under Assumption~\ref{assume:gap_accuracy_relation}, there exists a classical algorithm with run time
\begin{align}\label{eq:partition_trunc_small}
        \frac{R^2\abs{S}}{\epsilon}{\rm poly}\left[\left(\frac{\abs{S}}{\gamma^2\beta\epsilon[1-\beta/\beta^*]}\right)^{\log(\beta^*/\beta)}\right],
    \end{align}
where $\beta=\Delta^{-1}\ln(\gamma^{-2}\epsilon^{-1})$ and $\beta^*=(2e^2\mathfrak{d}(\mathfrak{d}+1))^{-1}$ satisfies $\beta<\beta^*$, solves the GSEE problem in Definition \ref{def:gsee} subjects to 
\begin{eqnarray}
    \epsilon> \epsilon^*=2e^2\mathfrak{d}(\mathfrak{d}+1).
\end{eqnarray}
\end{theorem}
\begin{proof}
The cluster expansion restricts the maximal imaginary time we can reach is given by $\beta< \beta^*=(2e^2\mathfrak{d}(\mathfrak{d}+1))^{-1}$. Subsequently, from our choice of $\beta$ given by Lemma \ref{lemma:beta_value} and Assumption \ref{assume:gap_accuracy_relation} on the relationship between the gap and accuracy, we know that $\beta\epsilon\geq 1$. The minimal error that we can reach is given by $(\beta^*)^{-1}$.

Next, from the RQITE algorithm given by Theorem~\ref{theorem:gsee_main_theorem}, we know that to achieve an $\epsilon$-accuracy estimation on the ground state energy, an $\epsilon^\prime=\gamma^2 \beta\epsilon$ accuracy is required for estimating the partition function. Subsequently, we invoke Theorem \ref{thm:partition_cluster} to estimate the partition function using cluster expansion, in which case, we substitute in the $\epsilon^\prime$. Finally, an overall $\epsilon^{-1}$ factor is accounted for because we need to evaluate the partition function for $\mathcal{O}(\epsilon^{-1})$ different $x$.
\end{proof}

It should be remarked that because of the limited accuracy we can achieve, Assumption \ref{assume:gap_accuracy_relation} also sets constraints on the spectral gap for the quantum system. This renders the dequantization algorithm rather restrictive. To circumvent such a conundrum, in the next section, we identify certain situations where analytic continuation could be leveraged to extend the results to arbitrary constant accuracy.

\subsection{Analytic regions for initial state with large enough overlap}
Now, we extend the dequantization of the GSEE algorithm to arbitrary constant accuracy by utilizing tools of analytic continuation to extend $\beta$ to an arbitrary constant in the cluster expansion. As mentioned at the beginning of Sec.~\ref{sec:deq}, we must identify an analytic region for the target function given by Eq.~\eqref{eq:target}. This is equivalent to determining the partition function's zero-free region, a conundrum to overcome in general settings. To this end, we consider the scenario where the initial state has a large enough overlap with the ground state, the partition function is guaranteed to be non-zero for $\beta$ around a disk in the complex plane, for which we have the following result.
\begin{theorem}[Zero-free region of the partition function]\label{thm:zero-free}
For the GSEE problem defined in Definition \ref{def:gsee}, when the initial-state overlap with the ground state satisfies $|\langle \psi_I| \psi_0\rangle|\geq \gamma= \frac{1}{\sqrt{2}},$
we have that the partition function $D_\beta(H)$ is zero-free for 
\begin{eqnarray}\label{eq:zero_free_condition}
    \rm{Re}(\beta)> 0,
\end{eqnarray}
where $\Delta$ is the energy gap of the Hamiltonian.
\end{theorem}
\begin{proof}
We first expand the partition function in the eigenbasis of the Hamiltonian:
\begin{eqnarray}\label{eq:zero_pf1}
\begin{split}
    D_\beta(H)&=\bra{\psi_I}e^{-\beta H} \ket{\psi_I}\\
    &=\sum_{i=0}^{N-1}p_i e^{-\beta E_i}\\
    &=e^{\beta E_0}\sum_{i=0}^{N-1}p_i e^{-\beta (E_i-E_0)}=:e^{\beta E_0}S_\beta(H),
\end{split}
\end{eqnarray}
where in the last line we have extracted the ground-state energy out and denoted the summation part as $S_\beta(H)$. We observe that the term $|e^{\beta E_0}|$ approaches zero only in the limit that $|\Re(\beta)|\rightarrow \infty$ so that zero points of the partition function are determined by the second part in the last line of Eq.~\eqref{eq:beta_proof_1}. For the second part, we observe that
\begin{eqnarray}\label{eq:eq:beta_proof_2}
\begin{split}
    |S_\beta(H)|&=\left|p_0+\sum_{i=1}^{N-1}p_i e^{-\beta (E_i-E_0)}\right|\\
    &\geq \left|p_0 +(1-p_0)e^{-\beta \Delta}\right|.
\end{split}
\end{eqnarray}
In the first line, we have singled out the contribution of the ground state, and we explain the results in the second line in the following. Because $p_0\geq0.5=\gamma^2$, which means the ground-state contribution dominants in the initial state, $|S_\beta(H)|$ becomes zero only when the excited-state contribution cancels out with the ground-state. It then follows that because $\rm{Re}(\beta)>0$ so that $|p_i e^{-\beta(E_i-E_0)}|\leq p_i,i\neq 0$, the excited-state elements in the summation are non-increasing with their energy separation with the ground state. We then identify the worst-case scenario, where the initial-state component that is orthogonal to the ground-state concentrates to the first excited state, i.e., $p_1=1-p_0$. Finally, by letting $p_0 +(1-p_0)e^{-\beta \Delta}=0$, we get
\begin{eqnarray}\label{eq:zero-free-final}
    \beta = \frac{1}{\Delta}\left(\pm i \pi+\ln{\frac{1-p_0}{p_0}}\right),
\end{eqnarray}
Then, by noting that $\rm{Re}(\beta)=\ln{\frac{1-p_0}{p_0}}\leq0$ since $p_0\geq0.5$, we find the zero point with the largest real part is not present in the right half complex plane. This completes the proof.
\end{proof}

Before we close this section, we note that when $p_0=0.5$, the zero point given by Eq.~\eqref{eq:zero-free-final} becomes $\frac{\pm i\pi}{\Delta}$. This suggests the radius of the analytic region of the logarithm of the partition function has the lower-bounded value $\Omega\left(\frac{\pi}{\Delta}\right)$. A similar result is obtained in [Theorem 9, Ref.~\cite{mann2024algorithmic}], where a concrete instance is established by the authors, and the condition is proved to be asymptotically optimal. Our construction here, on the other hand, could be seen as a non-constructive proof for further utility.


\subsection{Classical algorithm for solving the GSEE problem with general constant accuracy}
We provide details on the classical method for dequantization of the RQITE algorithm for arbitrary constant accuracy. Our construction follows from the deduction in Sec.~\ref{sec:analytic_continuation} and \ref{sec:arbitrary_complexity} for the cluster expansion.

A central difference between the construction here and the scenario where only limited accuracy can be achieved in Sec.~\ref{sec:deq_limit} is that here we apply a similarity-transformed Hamiltonian $H^\prime=U^\dagger H U$, where $\ket{\psi_c}=U\ket{0^n}$ is the semi-classical guiding state. As we assumed that the unitary $U$ prepares the semi-classical state of constant depth $d$, $H^\prime$ remains local. We denote the maximum degree of the interaction graph related to $H^\prime$ as $\mathfrak{d}^\prime$. To see how the degree changes from $\mathfrak{d}$ to $\mathfrak{d}^\prime$, we focus on two cases: (i) $D$-dimensional geometrically local quantum lattices, where the quantum circuit $U$ is accordingly assumed to be geometrically local, i.e.~brick wall quantum circuit, where the locality $k\mapsto d^Dk$ (ii) All-to-all interaction is allowed, yet each quantum gate of $U$ is only allowed to act on a constant number of qubits, where the locality $k\mapsto k2^d$. In both scenarios, we may upper bound the modified maximum degree $\mathfrak{d}^\prime$ via using Eq.~\ref{Eq: maximumdegree}.


Now, using the classical algorithm provided in Sec.~\ref{sec:arbitrary_complexity}, we have the following result.
\begin{theorem}[Dequantization of the RQITE algorithm with arbitrary accuracy]\label{thm:deq_rqiteapp}
Let $\ket{\psi_c}=U\ket{0^n}$ be the semi-classical state. For a targeted local Hamiltonian $H$, assume that $\ket{\psi_c}$ has a lower-bounded overlap $1/\sqrt{2}$ with the ground state of $H$. Let the locality of the similarity-transformed Hamiltonian $H^\prime=U^\dagger HU$ be denoted as $\mathfrak{d}^\prime$. Then, under Assumption~\ref{assume:gap_accuracy_relation}, there exists a classical algorithm that solves the GSEE problem in Definition \ref{def:gsee} with a run time
\begin{align}\label{eq:arbitrary_final_cost}
        \left[\frac{e^{2\pi\beta/\beta^*}}{\beta\epsilon^2}\rm{poly}(\abs{S})\right]^{ e^{2\pi\beta/\beta^*}},
    \end{align}
where $\beta=\Delta^{-1}\ln(\gamma^{-2}\epsilon^{-1})$, and $\beta^*=(2e^2\mathfrak{d}^\prime(\mathfrak{d}^\prime+1))^{-1}$.
\end{theorem}
\begin{proof}
The classical simulation algorithm follows from Sec.~\ref{sec:arbitrary_complexity}. That is we first obtain the similarity-transformed Hamiltonian $H^\prime$ so that the locality of the Hamiltonian is changing from $\mathfrak{d}$ to $\mathfrak{d}^\prime$. It should be noted that $\mathfrak{d}^\prime$ remains a constant as long as $\mathfrak{d}$ is a constant. Then, we apply the cluster expansion in combination with the analytic continuation to evaluate each $D_\beta(H^\prime-x)$ in the RQITE algorithm. To do so, we output an $\epsilon^\prime$-accuracy estimation $\log\hat{f}(x)$ to $\log f(x)=\log\left(D_{\beta\phi(1)}(H-x)\right)=\log\left(\langle0^n|e^{-\beta(H^{\prime}-x)}|0^n\rangle\right)$ using the cluster expansion algorithm described in Sec.~\ref{sec:arbitrary_complexity} with the cost given by Eq.~\eqref{eq:analytic_cost} in Theorem \ref{thm:partition_analytic_continuation}. It then implies that 
\begin{eqnarray}
    e^{-\epsilon^\prime}\hat{f}(x) \leq f(x) \leq e^{\epsilon^\prime} \hat{f}(x),
\end{eqnarray}
where $0\leq f(x), \hat{f}(x)\leq 1$ since $0\leq\langle0^n|e^{-\beta(H^{\prime}-x)}|0^n\rangle\leq 1$, for $\beta\in \mathbb{R}$ because $e^{-\beta(H^{\prime}-x)}$ is positive semi-definite and trace non-increasing. We thus take $\hat{f}(x)$ to be non-negative and all negative output can be rounded up to zero. This gives us that 
\begin{eqnarray}
    (1-e^{-\epsilon^\prime}) \leq |f(x)-\hat{f}(x)| \leq (e^{\epsilon^\prime}-1) .
\end{eqnarray}
For $(e^{\epsilon^\prime}-1)$, we use the trick that $e^{x}-1=\sum_{j=1}^\infty \frac{x^j}{j!}\leq x\sum_{j=0}^\infty \frac{x^j}{j!}=xe^x$ for sufficiently small $x$ and vice versa for $(1-e^{-\epsilon^\prime})$. Now, let $e^{\epsilon^\prime}-1\leq \epsilon^\prime e^{\epsilon^\prime}\leq \varepsilon$, which yields $\epsilon^\prime=W_0(\varepsilon)$ such that $W_0(x)$ is the Lambert W function. Finally, we take for convenience $\epsilon^\prime=\varepsilon/2$, which trivially satisfies the above analysis. To summarize, to fulfill an additive-error estimation in the partition function, it suffices to estimate the logarithm of the partition function with the accuracy doubled.

For dequantization of the RQITE algorithm, we thus set the accuracy in estimating the partition function to be $\varepsilon=2\gamma^2\beta\epsilon$, $\gamma^2=0.5$ and $\beta$ given by Eq.~\eqref{eq:beta_value} and invoke Theorem \ref{thm:partition_analytic_continuation}. Finally, we remark that a $\epsilon^{-1}$ overhead is introduced because there are $\mathcal{O}(\epsilon^{-1})$ of partition function we need to evaluate. This gives the final result in Eq.~\eqref{eq:arbitrary_final_cost}.
\end{proof}

\end{document}